\pdfoutput=1

\documentclass[acmsmall,screen]{acmart}

\AtBeginDocument{%
  \providecommand\BibTeX{{%
    \normalfont B\kern-0.5em{\scshape i\kern-0.25em b}\kern-0.8em\TeX}}}

\setcopyright{rightsretained}
\acmPrice{}
\acmDOI{10.1145/3434305}
\acmYear{2021}
\copyrightyear{2021}
\acmSubmissionID{popl21main-p157-p}
\acmJournal{PACMPL}
\acmVolume{5}
\acmNumber{POPL}
\acmArticle{24}
\acmMonth{1}

\makeatletter
\def\@copyrightpermission{
  \\
  \href{https://creativecommons.org/licenses/by/4.0/}{\includegraphics[width=1.6cm]{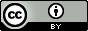}} \\
  This work is licensed under a \href{https://creativecommons.org/licenses/by/4.0/}{Creative 
  Commons Attribution 4.0 International License}}
\makeatother

\usepackage[T1]{fontenc}
\usepackage[utf8]{inputenc}

\usepackage{amsmath}
\usepackage{bbold} % alternative blackboard letters (amssymb already provides a set of them)
\usepackage{graphicx}
\usepackage{hyperref}
\usepackage{listings}
\usepackage{mathpartir} % inference rules
\usepackage{stmaryrd}
\usepackage{mathtools} % to typeset equations in the appendix
\usepackage{upquote}
\usepackage[all]{xy}

\begin{document}

% !TEX root = paper.tex

% Any macro that is actually used should have a comment explaining what it is for.
% Please fight macro pollution and remove the macros that are not used.

\newcommand{\defeq}{\mathrel{\overset{\text{\tiny def}}{=}}} % Definitional equality

\newcommand{\pl}[1]{\textsc{#1}} % the name of a programming language

\newcommand{\lambdaAEff}{$\lambda_{\text{\ae}}$} % the name of the calculus

% BNF grammars
\newcommand{\bnfis}{\mathrel{\;{:}{:}{=}\ }}
\newcommand{\bnfor}{\mathrel{\;\big|\ \ }}

%%%%% Semantic concepts

%%% Sets

\newcommand{\One}{\mathbb{1}} % singleton set as denotation of unit type
\newcommand{\one}{\star} % canonical element of the singleton set
\newcommand{\Zero}{\mathbb{0}} % empty set as denotation of empty type

\newcommand{\Bool}{\mathbb{B}} % two-element set of booleans
\newcommand{\true}{\mathbf{true}} % constant true
\newcommand{\false}{\mathbf{false}} % constant false

\newcommand{\expto}{\Rightarrow} % set exponentiation
\newcommand{\lam}[1]{\lambda #1 \,.\,} % lambda abstraction
\newcommand{\pair}[2]{\langle #1 , #2 \rangle} % pairing

\newcommand{\lifted}[1]{#1_\bot} % lifting monad
\newcommand{\idte}[4]{\mathbf{ifdef}~#1~\mathbf{then}~#2 \mapsto #3~\mathbf{else}~#4} % test if element of a lifted set is defined (non-bottom) or not, and then use it in the then branch

\newcommand{\ite}[3]{\mathbf{if}~#1~\mathbf{then}~#2~\mathbf{else}~#3} % if-then-else used in semantic definitions

%%% Signatures

\newcommand{\Tree}[2]{\mathrm{Tree}_{#1}\left(#2\right)} % The tree algebra for an operation signature
\newcommand{\retTree}[1]{\mathsf{return}\,#1} % the inclusion of generators into trees

\newcommand{\opsym}[1]{\mathsf{#1}} % a custom operation symbol
\newcommand{\op}{\opsym{op}} % a generic operation symbol

\newcommand{\sig}{\Sigma} % the global signature of signal and interrupt names

\renewcommand{\o}{o} % effect annotation describing possible outgoing operations
\renewcommand{\i}{\iota} % effect annotation describing possible incoming operations

\newcommand{\opincomp}[2]{{\mathsf{#1}}\,{\tmkw{\downarrow}}\,#2} % action of incoming interrupt on computation types
\newcommand{\opincompp}[2]{{\mathsf{#1}}\,{\tmkw{\downarrow\downarrow}}\,#2} % action of a list of incoming interrupts on computation types

%%% Theories
\newcommand{\eq}{\mathrm{Eq}} % a set of equations

\newcommand{\FreeAlg}[2]{\mathrm{Free}_{#1}\left(#2\right)} % Free algebra for a signature generated by a set
\newcommand{\lift}[1]{#1^\dagger} % the Kleisli lifting of a map
\newcommand{\freelift}[1]{#1^\ddagger} % the lifting of a map induced by the free model property

\newcommand{\M}{\mathcal{M}} % a generic model for a theory
\newcommand{\Mcarrier}{\vert \mathcal{M} \vert} % the carrier of a generic model

\newcommand{\T}{T} % A generic monad

%%% Example effect theories

\newcommand{\sigget}{\mathsf{get}}
\newcommand{\sigset}{\mathsf{set}}

%%%%% Types

\newcommand{\at}{\mathbin{!}} % the ! sign, with proper spacing
\newcommand{\att}{\mathbin{!!}} % the !! sign, with proper spacing

%% Value types

\newcommand{\tysym}[1]{\mathsf{#1}}
\newcommand{\tybase}{\tysym{b}} % a base type
\newcommand{\tyunit}{\tysym{1}} % the unit ground type
\newcommand{\tyint}{\tysym{int}} % the integer ground type
\newcommand{\tystring}{\tysym{string}} % the integer ground type
\newcommand{\tylist}[1]{\tysym{list}~\tysym{#1}} % the list ground type
\newcommand{\tyempty}{\tysym{0}} % the empty ground type
\newcommand{\typrod}[2]{#1 \times #2} % product type
\newcommand{\tysum}[2]{#1 + #2} % sum type
\newcommand{\tyfun}[2]{#1 \to #2} % user function type
\newcommand{\typromise}[1]{\langle #1 \rangle} % type of promises

%% Computation types

\newcommand{\tycomp}[2]{#1 \at #2} % computation type

%% Process types

\newcommand{\tyrun}[3]{#1 \att (#2,#3)} % type of the run M process
\newcommand{\typar}[2]{#1 \mathbin{\tmkw{\vert\vert}}  #2} % type of parallel processes
\newcommand{\tyC}{C} % meta variable ranging over process types
\newcommand{\tyD}{D} % meta variable ranging over process types

%%%%% Display of source code in math mode

\newcommand{\tm}[1]{\mathsf{#1}} % the source code font
\newcommand{\tmkw}[1]{\tm{\color{keywordColor}#1}} % source code keyword, colored

\newcommand{\tmpromise}[1]{\langle #1 \rangle} % completed promise

\newcommand{\tmconst}[1]{\tm{#1}}
\newcommand{\tmunit}{()} % the element of the unit type
\newcommand{\tmpair}[2]{( #1 , #2 )} % ordered pair
\newcommand{\tminl}[2][]{\tmkw{inl}_{#1}\,#2} % left injection
\newcommand{\tminr}[2][]{\tmkw{inr}_{#1}\,#2} % right injection
\newcommand{\tmfun}[2]{{\mathop{\tmkw{fun}}}\; (#1) \mapsto #2} % function abstraction
\newcommand{\tmfunano}[2]{{\mathop{\tmkw{fun}}}\; #1 \mapsto #2} % function abstraction (no type annotation expected)
\newcommand{\tmapp}[2]{#1\,#2} % application

\newcommand{\tmreturn}[2][]{\tmkw{return}_{#1}\, #2} % pure computation
\newcommand{\tmlet}[3]{\tmkw{let}\; #1 = #2 \;\tmkw{in}\; #3} % let-binding
\newcommand{\tmletrec}[5][]{\tmkw{let}\;\tmkw{rec}\; #2\; #3 #1 = #4 \;\tmkw{in}\; #5} % recursive definitions

\newcommand{\tmop}[4]{\tm{#1}\;(#2, #3. #4)} % operation call
\newcommand{\tmopin}[3]{\tmkw{\downarrow}\, \tm{#1}\,(#2, #3)} % incoming interrupt
\newcommand{\tmopout}[3]{\tmkw{\uparrow}\,\tm{#1}\, (#2, #3)} % outgoing signal
\newcommand{\tmopoutbig}[3]{\tmkw{\uparrow}\,\tm{#1}\, \big(#2, #3\big)} % outgoing signal with big brackets
\newcommand{\tmopoutgen}[2]{\tmkw{\uparrow}\,\tm{#1}\, #2} % generic variant of outgoing signal

\newcommand{\tmmatch}[3][]{\tmkw{match}\;#2\;\tmkw{with}\;\{#3\}_{#1}} % match statement

\newcommand{\tmawait}[3]{\tmkw{await}\;#1\;\tmkw{until}\;\tmpromise{#2}\;\tmkw{in}\;#3} % awaiting for a promise to be completed

\newcommand{\tmwith}[5]{\tmkw{promise}\; (\tm{#1}\; #2 \mapsto #3)\; \tmkw{as}\; #4\; \tmkw{in}\; #5} % interrupt hook

\newcommand{\tmrun}[1]{\tmkw{run}\; #1} % running a computation as a process
\newcommand{\tmpar}[2]{#1 \mathbin{\tmkw{\vert\vert}} #2} % parallel composition of processes

%%% Operational semantics

\newcommand{\reduces}{\leadsto} % small-step reduction
\newcommand{\tyreduces}{\rightsquigarrow} % reduction of process types

\newcommand{\E}{\mathcal{E}} % evaluation context for computations
\renewcommand{\H}{\mathcal{H}} % signal hoisting context
\newcommand{\F}{\mathcal{F}} % evaluation context for processes

%%% Typing rules

\newcommand{\types}{\vdash} % typing judgement
\newcommand{\of}{\mathinner{:}} % the colon in a typing judgement

\newcommand{\sub}{\sqsubseteq} % subtyping relation

\newcommand{\coopinfer}[3]{\inferrule*[Lab={\color{rulenameColor}#1}]{#2}{#3}}

%%% Meta-theory

\makeatletter
\newcommand{\hourglass}{}                  % hourglass symbol for classifying temporarity blocked computations
\DeclareRobustCommand{\hourglass}{\mathrel{\mathpalette\hour@glass\relax}}

\newcommand\hour@glass[2]{%
  \vcenter{\hbox{%
    \rotatebox[origin=c]{90}{\scalebox{0.8}{$\m@th#1\bowtie$}}%
  }}%
}
\makeatother

\newcommand{\awaiting}[2]{#1 \hourglass #2} % computations blocked on awaiting a particular promise variable to be fulfilled

\newcommand{\CompResult}[2]{\mathsf{CompRes}\langle#1 \,\vert\, #2\rangle} % top-level result forms of individual computations
\newcommand{\RunResult}[2]{\mathsf{RunRes}\langle#1 \,\vert\, #2\rangle} % local (under-signal) result forms of individual computations

\newcommand{\Result}[2]{\mathsf{Res}\langle#1 \,\vert\, #2\rangle} % top-level result forms of computations

\newcommand{\ProcResult}[1]{\mathsf{ProcRes}\langle #1 \rangle} % top-level result forms of parallel processes
\newcommand{\ParResult}[1]{\mathsf{ParRes}\langle #1 \rangle} % intermediate result forms of parallel processes

%%% Maths

\newcommand{\cond}[3]{\mathsf{if}\;#1\;\mathsf{then}\;#2\;\mathsf{else}\;#3} % single line conditional

\newcommand{\carrier}[1]{\vert #1 \vert} % carrier of a cpo
\newcommand{\order}[1]{\sqsubseteq_{#1}} % partial order of a cpo
\newcommand{\lub}[1]{\bigsqcup_n \langle #1 \rangle} % least upper bound of an omega-chain

\newcommand{\Pow}[1]{\mathcal{P}(#1)} % powerset
\newcommand{\sem}[1]{[\![#1]\!]} % semantic bracket

\makeatletter
\providecommand*{\cupdot}{%     % disjoint union of sets
  \mathbin{%
    \mathpalette\@cupdot{}%
  }%
}
\newcommand*{\@cupdot}[2]{%
  \ooalign{%
    $\m@th#1\cup$\cr
    \hidewidth$\m@th#1\cdot$\hidewidth
  }%
}
\makeatother

%%% Redex highlighting

\definecolor{redexColor}{rgb}{0.83, 0.83, 0.83} % the color of highlighted redexes
\newcommand{\highlightgray}[1]{{\setlength{\fboxsep}{1.5pt}\colorbox{redexColor}{$#1$}}} % highlight redexes with gray(ish) background
\newcommand{\highlightwhite}[1]{{\setlength{\fboxsep}{1.5pt}\colorbox{white}{$#1$}}} % highlight redexes with white background

%% AUTHORS
\author{Danel Ahman}
\orcid{0000-0001-6595-2756}
\email{danel.ahman@fmf.uni-lj.si}
\author{Matija Pretnar}
\orcid{0000-0001-7755-2303}
\email{matija.pretnar@fmf.uni-lj.si}
\affiliation{%
  \institution{University of Ljubljana}
  \department{Faculty of Mathematics and Physics}
  \streetaddress{Jadranska 21}
  \city{Ljubljana}
  \postcode{SI-1000}
  \country{Slovenia}
}

%%%
%%% The code below is generated by the tool at http://dl.acm.org/ccs.cfm.
%%% Please copy and paste the code instead of the example below.
%%%


%%%
%%% Keywords. The author(s) should pick words that accurately describe
%%% the work being presented. Separate the keywords with commas.

%% TITLE
\title{Asynchronous Effects}

%% ABSTRACT
\begin{abstract}
We explore asynchronous programming with algebraic effects. 
We complement their conventional synchronous treatment by showing how to naturally 
also accommodate asynchrony within them, namely, by decoupling the execution of 
operation calls into signalling that an operation’s implementation needs to be executed, and 
interrupting a running computation with the operation’s result, to which the computation can 
react by installing interrupt handlers. We formalise these ideas in a small core calculus, 
called \lambdaAEff. We demonstrate the flexibility of \lambdaAEff~using examples ranging 
from a multi-party web application, to preemptive multi-threading, to remote function 
calls, to a parallel variant of runners of algebraic effects. In addition, the paper is accompanied by 
a formalisation of \lambdaAEff's type safety proofs in \pl{Agda}, and a prototype implementation 
of \lambdaAEff~in \pl{OCaml}.
\end{abstract}

\maketitle

%%%%%%%%%%%%%%%%%%%%%%%%%%%%%%%%%%%%%%%%%%%%%%%%%%%%

%% CODE SNIPPETS TYPESETTING

\definecolor{codegreen}{rgb}{0,0.6,0}
\definecolor{codegray}{rgb}{0.5,0.5,0.5}
\definecolor{codepurple}{rgb}{0.58,0,0.82}
\definecolor{backcolour}{rgb}{0.95,0.95,0.92}

\definecolor{keywordColor}{rgb}{0.0,0.0,0.5} % the color of language keywords
\definecolor{rulenameColor}{rgb}{0.5,0.5,0.5} % the color of rule names

%% VARIOUS SETTINGS

\citestyle{acmauthoryear}

\newcommand{\sref}[2]{\hyperref[#2]{#1~\ref{#2}}} % workaround for autoref using Theorem for Lemmas and Propositions
\newcommand{\srefcase}[3]{\hyperref[#2]{#1~\ref{#2}~(#3)}} % workaround for autoref using Theorem for Lemmas and Propositions (linking to a particular case of the Lemma)

\def\sectionautorefname{Section}
\def\subsectionautorefname{Section}
\def\subsubsectionautorefname{Section}

\def\lstlanguagefiles{aeff}
\lstset{language=aeff,upquote=true}
\let\ls\lstinline

%%%%%%%%%%%%%%%%%%%%%%%%%%%%%%%%%%%%%%%%%%%%%%%%%%%%

%% PAPER CONTENTS

% !TEX root = paper.tex

\section{Introduction}

Effectful programming abstractions are at the heart of many modern general-purpose 
programming languages.
They can increase expressiveness by giving 
access to first-class continuations, but often simply help users to write 
cleaner code, e.g., by avoiding having to manage a program's memory explicitly in state-passing style, 
or getting lost in callback hell while programming asynchronously.

An increasing number of language designers and programmers are starting to 
embrace \emph{algebraic effects}, 
where one uses algebraic operations \cite{Plotkin:NotionsOfComputation} and 
effect handlers \cite{Plotkin:HandlingEffects} to uniformly and user-definably 
express a wide range of effectful behaviour, 
ranging from basic examples such as state, rollbacks, exceptions, 
and nondeterminism \cite{Bauer:AlgebraicEffects}, to advanced applications 
in concurrency \cite{Dolan:MulticoreOCaml} and statistical probabilistic programming 
\cite{Bingham:Pyro}, and even quantum computation \cite{Staton:AlgEffQuantum}.

While covering many examples, the conventional treatment of 
algebraic effects is \emph{synchronous} by nature. In it 
effects are invoked by placing operation calls in one's code, 
which then propagate outwards until they trigger the actual effect, finally yielding 
a result to the rest of the computation that has been \emph{waiting} the whole 
time. While blocking the computation is indeed sometimes needed, e.g., 
in the presence of general effect handlers that can execute their continuation any 
number of times, it forces all uses of algebraic effects to be synchronous, even when this
is not necessary, e.g., when the effect involves executing 
a remote query to which a response is not needed (immediately).

Motivated by the recent interest in the combination of
asynchrony and algebraic effects \cite{Leijen:AsyncAwait,Dolan:MulticoreOCaml}, 
we explore what it takes (in terms of
language design, safe programming abstractions, and a 
self-contained core calculus) to accompany the
synchronous treatment of algebraic effects with
an \emph{asynchronous} one. At the heart of our approach is the 
decoupling of the execution of operation calls
into \emph{signalling} that some implementation of an operation needs to be executed, and \emph{interrupting} a
running computation with its result, to which the computation can react by
installing \emph{interrupt handlers}. Importantly, we show that our 
approach is flexible enough that not all signals need to have a
corresponding interrupt, and vice versa, allowing us to also model 
\linebreak % so as to force typesetting not to jump around
\emph{spontaneous behaviour}, such as a
user clicking a button or the environment preempting a thread.

While we are not the first ones to work on asynchrony for algebraic effects, 
the prior work in this area (in the context of general effect handlers) has 
achieved it by \emph{delegating} the actual asynchrony to the respective language backends 
\cite{Leijen:AsyncAwait,Dolan:MulticoreOCaml}. In contrast, in this paper 
we demonstrate how to capture the combination of 
asynchrony and algebraic effects in a \emph{self-contained} core calculus. 
It is important to emphasise that our aim is not to replace general effect handlers,  
but instead to \emph{complement} them with robust primitives 
tailored to asynchrony---our proposed approach is algebraic by design, so as 
to be ready for future extensions with general effect handlers.

\paragraph{Paper structure}
In \autoref{sec:overview}, we give a high-level overview of our approach to 
asynchrony for algebraic effects. 
In \autoref{sec:basic-calculus:computations} 
and \ref{sec:basic-calculus:processes}, we distil our ideas into a core calculus, \lambdaAEff, 
equipped with a small-step semantics, a type-and-effect system, and proofs of 
type safety. In \autoref{sec:applications}, 
we show \lambdaAEff~in action on examples 
such as preemptive multi-threading, remote function calls, and a parallel variant of runners 
of algebraic effects. We conclude, and discuss related and future work in \autoref{sec:conclusion}.

\paragraph{Code}
The paper is accompanied by a \emph{formalisation} of \lambdaAEff's type safety proofs 
in \pl{Agda} \cite{ahman20:AeffAgda}, and a \emph{prototype implementation} of \lambdaAEff~in 
\pl{OCaml}, called \pl{{\AE}ff} \cite{pretnar20:AEff}. For ease of use, we provide them both also as 
a single virtual machine image \cite{AhmanPretnar20:Artefact}.

In the \pl{Agda} formalisation, we consider only well-typed syntax of a 
variant of \lambdaAEff~in which the subsumption rule manifests as an explicit coercion, so as to make working with 
de Bruijn indices less painful. 
Meanwhile, the \pl{{\AE}ff} implementation provides an interpreter and 
a simple typechecker, but it does not 
yet support inferring and checking effect annotations. In addition, \pl{{\AE}ff} provides   
a web interface that allows users to enter their programs and interactively click through 
their executions.
\pl{{\AE}ff} also comes with implementations of all the examples we present in this paper.

Separately, \citet{Poulson:AsyncEffectHandling} has shown how to implement \lambdaAEff~
in \pl{Frank} \cite{Convent:DooBeeDooBeeDoo}.

% !TEX root = paper.tex

\section{Asynchronous Effects, by Example}
\label{sec:overview}

We begin with a high-level overview of how we accommodate asynchrony within algebraic effects.

\subsection{Conventional Algebraic Effects Are Synchronous by Nature}

We first recall the basic ideas of programming with algebraic effects, 
illustrating that their traditional treatment is synchronous by nature.
For an in-depth overview, we refer to the tutorial by \citet{Pretnar:Tutorial}, and to the 
seminal papers of the field \cite{Plotkin:NotionsOfComputation,Plotkin:HandlingEffects}.

In this algebraic treatment, sources of computational effects are modelled using signatures 
of \emph{operation symbols} $\op : A_\op \to B_\op$. For instance, one models 
$S$-valued state using two operations, $\sigget : \tyunit \to S$ and $\sigset : S \to \tyunit$; 
and $E$-valued exceptions using a single operation $\opsym{raise} : E \to \tyempty$.

Programmers can then invoke the effect that an 
$\op : A_\op \to B_\op$ models by placing an \emph{operation call} $\tmop {op} V y M$ in their code. Here, the 
parameter value $V$ has type $A_\op$, and the variable $y$, which is bound in the continuation $M$, has type $B_\op$.
For instance, for $\sigset$, the parameter $V$ would be 
the new value of the store, and for $\sigget$, the variable $y$ would be bound to the current value of the store.

A program written in terms of operation calls is by itself just an inert piece of code. To 
execute it, programmers have to provide \emph{implementations} for the operation 
calls appearing in it. The idea is that an implementation of $\tmop {op} V y M$ takes $V$ as its input, 
and its output gets bound to $y$.
For instance, this could take the form of defining a suitable effect handler 
\cite{Plotkin:HandlingEffects}, but could also be given by calls  
to runners of algebraic effects \cite{Ahman:Runners}, or simply by invoking some  
(default) top-level (native) implementation.
What is important is that some pre-defined piece of code $M_\op[V/x]$
gets executed in place of every operation call $\tmop {op} V y M$.

Now, what makes the conventional treatment of algebraic effects \emph{synchronous} is 
that the execution of an operation call $\tmop {op} V y M$ \emph{blocks} until some implementation 
of $\op$ returns a value $W$ to be bound to $y$, so that 
the execution of the continuation $M[W/y]$ could proceed \cite{Kammar:Handlers,Bauer:EffectSystem}. 
Conceptually, this kind of blocking behaviour can be illustrated as
\begin{equation}
\begin{gathered}
\label{eq:syncopcall}
\xymatrix@C=1.25em@R=0.85em@M=0.5em{
& M_\op[V/x] \ar@{}[r]|{\mbox{$\Large{\leadsto^{\!*}}$}} & \tmreturn W \ar[d]
\\
\dots \ar@{}[r]|>>>{\mbox{$\Large{\leadsto}$}} & \tmop {op} V y M \ar[u] & M[W/y] \ar@{}[r]|<<<{\mbox{$\Large{\leadsto}$}} & \dots
}
\end{gathered}
\end{equation}
where $\tmreturn W$ is a computation that causes no effects and simply returns the given value $W$.

While blocking the rest of the computation is needed in the presence of 
general effect handlers that can execute their continuation  any number 
of times, it forces all uses of algebraic effects to be synchronous, even 
when this is not necessary, e.g., when the effect in question involves 
executing a remote query to which a response is not needed immediately, 
or sometimes never at all.

In the rest of this section, we describe how we decouple the invocation of 
an operation call from the act of receiving its result, and how we give 
programmers a means to block execution only when it is necessary. 
While we end up surrendering some of effect handlers' generality, 
such as having access to the continuation that captures the rest of the 
computation to be handled, then in return we get a natural and robust formalism for 
asynchronous programming with algebraic effects.

\subsection{Outgoing Signals and Incoming Interrupts}
\label{sec:overview:signals}

We begin by observing that the execution of an operation call $\tmop {op} V y M$, 
as shown in (\ref{eq:syncopcall}), consists of \emph{three distinct phases}: (i) signalling that an 
implementation of $\op$ needs to be executed with parameter $V$ (the up-arrow), (ii) executing 
this implementation (the horizontal arrow), and (iii) interrupting the blocking of $M$ with a value $W$ 
(the down-arrow). In order to overcome the unwanted side-effects of blocking execution on every operation call, 
we shall naturally decouple these phases into separate programming concepts, allowing the execution of 
$M$ to proceed even if (ii) has not yet completed and (iii) taken place. In particular, we 
decouple an operation call into issuing an \emph{outgoing signal}, 
written $\tmopout{\op}{V}{M}$, and receiving an \emph{incoming interrupt}, written $\tmopin{\op}{W}{M}$.

It is important to note that while we have used the execution of operation calls  
to motivate the introduction of signals and interrupts as programming concepts, \emph{not all issued signals need to have a corresponding 
interrupt response}, and \emph{not all interrupts need to be responses to issued signals}, 
allowing us to also model spontaneous behaviour, such as the environment preempting a thread.

When \emph{issuing a signal} $\tmopout{\op}{V}{M}$, the value $V$ is a \emph{payload}, such as a location to be looked up or a 
message to be displayed, aimed at whoever is listening for the given signal. We use the $\tmkw{\uparrow}$-notation to indicate that signals issued in sub-computations propagate outwards---in this sense signals behave just like conventional 
operation calls. However, signals crucially differ from conventional operation calls in that no additional variables 
are bound in the continuation $M$, making it naturally possible to continue executing $M$ straight after the signal has been issued, e.g., as depicted below:
\vspace{-3ex}
\[
\xymatrix@C=1.25em@R=1.25em@M=0.5em{
& &
\\
\dots \ar@{}[r]|<<<{\mbox{$\Large{\leadsto}$}} & \tmopout {op} V M \ar[u]^{\op\, V} \ar@{}[r]|<<<{\mbox{$\Large{\leadsto}$}} & M \ar@{}[r]|<<<{\mbox{$\Large{\leadsto}$}} & \dots
}
\]

\newcommand{\client}{M_{\text{feedClient}}}

As a \emph{running example}, consider a computation $\client$, which lets a user scroll through a seemingly infinite feed, e.g., by repeatedly clicking a ``next page'' button.
For efficiency, $\client$ does not initially cache all the data, but instead requests a 
new batch of data each time the user is nearing the end of the cache. To communicate with the outside world, $\client$ can issue a signal
\[
  \tmopout{\opsym{request}}{\mathit{cachedSize} + 1}{\client}
\]
to request a new batch of data starting from the end of the current cache, or a different signal
\[
  \tmopout{\opsym{display}}{\mathit{message}}{\client}
\]
to display a message to the user. In both cases, the continuation \emph{does not wait} 
for an acknowledgement that the signal was received, 
but instead continues to provide a seamless experience to the user.
It is however worth noting that these signals differ in what $\client$ expects of them: 
to the $\opsym{request}$ signal, it expects a response at some future point in 
its execution, while it does not expect any response to the $\opsym{display}$ signal, 
illustrating that not every issued signal needs a response.

When the outside world wants to get the attention of a computation, be it in response to 
a signal or spontaneously, 
it happens by \emph{propagating an interrupt}~$\tmopin{\op}{W}{M}$ to the computation. 
Here, the value $W$ is again a payload, while $M$ is the computation receiving the interrupt.
It is important to note that unlike signals, interrupts are not triggered by the computation itself, 
but are instead issued by the \emph{outside world},
and can thus interrupt any sequence of evaluation steps,
e.g., as in
\vspace{-3ex}
\[
\xymatrix@C=1.25em@R=1.25em@M=0.5em{
&  \ar[d]^-{\op\, W}  &
\\
\dots \ar@{}[r]|<<<{\mbox{$\Large{\leadsto}$}} & M \ar@{}[r]|<<<{\mbox{$\Large{\leadsto}$}} & \tmopin {op} W M \ar@{}[r]|<<<{\mbox{$\Large{\leadsto}$}} & \dots
}
\]

In our running example, there are two interrupts of interest: 
$\tmopin{\opsym{response}}{\mathit{newBatch}}{M}$, which delivers new data to replenish the 
cache; and $\tmopin{\opsym{nextItem}}{\tmunit}{M}$, with which the user requests to see the next data 
item. In both cases, $M$ represents the state of $\client$ before the interrupt arrived.

We use 
the $\tmkw{\downarrow}$-notation to indicate that interrupts propagate inwards into sub-computations, 
trying to reach anyone listening for them, and only get discarded when they reach a $\tmkw{return}$. 
It is worth noting that programmers are not expected to write interrupts explicitly in their programs---instead, 
interrupts are usually induced by signals issued by other parallel processes, as explained next. 

\subsection{A Signal for the Sender Is an Interrupt to the Receiver}
\label{sec:overview:processes}

As noted above, the computations we consider do not evolve in isolation, instead they also communicate with 
the outside world, by issuing outgoing signals and receiving incoming interrupts.

We model the outside world by composing individual computations into \emph{parallel processes} $P, Q, \ldots$.
To keep the presentation clean and focussed on the asynchronous use of algebraic effects, we consider a very 
simple model of parallelism: a process is either one of the individual computations being run
in parallel, written $\tmrun M$, or the parallel composition of two processes, 
written $\tmpar P Q$.

\newcommand{\server}{M_{\text{feedServer}}}

To capture the signals and interrupts based interaction of processes,  
our operational semantics includes rules for \emph{propagating outgoing signals} from individual 
computations to processes, 
\emph{turning processes' outgoing signals into incoming interrupts} for their surrounding world, and
\emph{propagating incoming interrupts} from processes to individual computations.
For instance, in our running example, 
$\client$'s request for new data is executed as follows (with the active redexes highlighted):
\[
\begin{array}{r l}
  & \tmpar{\highlightgray{\tmrun (\tmopout{request}{V}{\highlightwhite{\client}})}}{\tmrun \server} 
  \\[0.5ex]
  \reduces & \highlightgray{\tmpar{(\tmopout{request}{V}{\highlightwhite{\tmrun \client}})}{\highlightwhite{\tmrun \server}}}
  \\[0.5ex]
  \reduces & \tmopoutbig{request}{V}{\tmpar{\tmrun \client}{\highlightgray{\tmopin{request}{V}{\tmrun {\highlightwhite{\server}}}}}}
  \\[0.5ex]
  \reduces & \tmopoutbig{request}{V}{\tmpar{\tmrun \client}{\tmrun (\tmopin{request}{V}{\server})}}
\end{array}
\]

Here, the first and the last reduction step respectively propagate signals outwards and 
interrupts inwards. The middle reduction step corresponds to what we call a \emph{broadcast rule}---it 
turns an outward moving signal in one of the processes into an inward moving interrupt for the process 
parallel to it, while continuing to propagate the signal outwards to any further parallel processes.

\subsection{Promising to Handle Interrupts}
\label{sect:overview:promising}

So far, we have shown that our computations can issue outgoing signals and receive incoming interrupts, and how 
these evolve when executing parallel processes, but we have not yet said 
anything about how computations can actually \emph{react} to incoming interrupts of interest. 

In order to react to incoming interrupts, our computations can install \emph{interrupt handlers}, written
\[
  \tmwith{op}{x}{M}{p}{N}
\]
that should be read as: ``we promise to handle a future interrupt named $\op$ using the computation  
$M$ in the continuation $N$, with $x$ bound to the payload of the interrupt''. Fulfilling this promise consists of executing $M$ and binding its result to the 
variable $p$ in $N$. This is captured by the reduction rule 
\[
  \tmopin{op}{V}{\tmwith{op}{x}{M}{p}{N}} \reduces \tmlet{p}{M[V/x]}{\tmopin{op}{V}{N}}
\]

It is worth noting two things: the interrupt handler is \emph{not reinstalled by default}, 
and the interrupt itself \emph{keeps propagating inwards} into the sub-computation $N$.
Regarding the former, 
programmers can selectively reinstall interrupt handlers when needed, 
by defining them suitably recursively, e.g., 
as we demonstrate in \autoref{sec:overview:runningexample}.
Concerning the latter, then in order to skip certain interrupt handlers for some $\opsym{op}$, one 
can carry additional data 
in $\opsym{op}$'s payload (e.g., a thread ID) and then condition the (non-)triggering of those interrupt 
handlers on this data, e.g., as we do in \autoref{sec:applications:guarder-handlers}.

Interrupts that do not match a given interrupt handler ($\op \neq \op'$) are simply propagated past it:
\[
  \tmopin{op'}{V}{\tmwith{op}{x}{M}{p}{N}} \reduces \tmwith{op}{x}{M}{p}{\tmopin{op'}{V}{N}}
\]

Interrupt handlers differ from operation calls in two important aspects.
First, they enable \emph{user-side post-processing} of received data, using $M$, 
while in operation calls the result is immediately bound in the continuation. Second, and more 
importantly, their semantics is \emph{non-blocking}. In particular, 
\[
N \reduces N' \qquad \text{implies} \qquad \tmwith{op}{x}{M}{p}{N} \reduces \tmwith{op}{x}{M}{p}{N'}
\]
meaning that the continuation $N$, and thus the whole computation, can make progress 
even though no incoming interrupt $\opsym{op}$ has been propagated to the computation 
from the outside world.

As the observant reader might have noticed, the non-blocking behaviour of interrupt handling 
means that our operational semantics has to work on \emph{open terms} because the variable $p$ can 
appear free in both $N$ and $N'$ above. However, it is important to note that $p$ is not an arbitrary variable, 
but in fact gets assigned a distinguished \emph{promise type} $\typromise X$ for some value type 
$X$---we shall crucially make use of this typing of $p$ in the proof of type safety for our \lambdaAEff-calculus (see \autoref{theorem:progress}).

\subsection{Blocking on Interrupts Only When Necessary}
\label{sec:overview:await}

As noted earlier, installing an interrupt handler means making a promise to handle a given 
interrupt in the future. To check that an interrupt has been received and handled, 
we provide programmers a means to selectively \emph{block execution}
and \emph{await} a specific promise to be fulfilled, written 
$\tmawait{V}{x}{M}$, where if $V$ has a promise type $\typromise X$, the variable $x$ bound in $M$ has type $X$.
Importantly, the continuation $M$ is executed only 
when the $\tmkw{await}$ is handed a \emph{fulfilled promise} $\tmpromise V$: 
\[
\tmawait{\tmpromise V}{x}{M} \reduces M[V/x]
\]

Revisiting our example of scrolling through a seemingly infinite feed,
$\client$ could use $\tmkw{await}$ to block until it has received an initial configuration, 
such as the batch size used by $\server$.

As the terminology suggests, this part of \lambdaAEff~is strongly influenced by existing work on 
\emph{futures and promises} \cite{Schwinghammer:Thesis} for structuring concurrent programs, and their use in modern languages, 
such as in \pl{Scala} \cite{Haller:Futures}. While prior work often models promises as writable, single-assignment 
references, we instead use the substitution of values for ordinary immutable variables (of distinguished promise type) 
to model that a promise gets fulfilled exactly once.

\subsection{Putting It All Together}
\label{sec:overview:runningexample}

Finally, we show how to implement our example of scrolling through a seemingly infinite feed.
For a simpler exposition, we allow ourselves access to mutable references, though the same can be  
achieved by rolling one's own state. 
Further, we use $\tmopoutgen {op} V$    
as a syntactic sugar for $\tmopout {op} V {\tmreturn \tmunit}$.

\subsubsection{Client}
\label{sec:overview:runningexample:client}

We implement the client computation $\client$ as the function \ls$client$ defined below. 
For presentation purposes, we split the definition of \ls$client$ between multiple code blocks.

First, the client sets up the initial values of the auxiliary references, 
issues a signal to the server asking for the data batch size that it uses, and then installs a corresponding
interrupt handler:
\begin{lstlisting}
let client () =
    let (cachedData , requestInProgress , currentItem) = (ref [] , ref false , ref 0) in
    send batchSizeRequest ();
    promise (batchSizeResponse batchSize |-> return <<batchSize>>) as batchSizePromise in
\end{lstlisting}

While the server is asynchronously responding to the batch size request, the client
sets up an auxiliary function \ls$requestNewData$, which it later uses to request new data from the server:
\begin{lstlisting}
    let requestNewData offset =
        requestInProgress := true;
        send request offset;
        promise (response newBatch |->
            cachedData := !cachedData @ newBatch;
            requestInProgress := false; return <<()>>
        ) as _ in return ()
    in
\end{lstlisting}
Here, the client first sets a flag indicating that a new data request is in process, 
then issues a $\opsym{request}$ signal to the server, and finally installs an 
interrupt handler that updates the cache 
once the $\opsym{response}$ interrupt arrives.
Note that the client does not block while awaiting new data, instead it continues executing, notifying 
the user to wait and try again once the cache is empty (see below).

Then, the client sets up its main loop, which is a simple recursively defined interrupt handler:
\begin{lstlisting}
    let rec clientLoop batchSize =
        promise (nextItem () |->
            let cachedSize = length !cachedData in
            (if (!currentItem > cachedSize - batchSize / 2) && (not !requestInProgress) then
                 requestNewData (cachedSize + 1)
             else
                 return ());
            (if !currentItem < cachedSize then
                 send display (toString (nth !cachedData !currentItem));
                 currentItem := !currentItem + 1
             else  
                 send display "please wait a bit and try again");
            clientLoop batchSize
        ) as p in return p
    in
\end{lstlisting}
In it, the client listens for a $\opsym{nextItem}$ interrupt from the user to display more data.
Once the interrupt arrives, the client checks if its cache is becoming empty---if so, 
it uses the $\opsym{requestNewData}$ function to 
request more data from the server. Next, if there is still some 
data in the cache, the client issues a signal to display the next data item to the user. 
If however the cache is empty, the client issues a signal 
to display a waiting message to the user. The client then simply recursively reinvokes itself.

As a last step of setting itself up, the client blocks until the server has responded 
with the batch size it uses, after which the client starts its main loop 
with the received batch size as follows:
\begin{lstlisting}
    await batchSizePromise until <<batchSize>> in clientLoop batchSize
\end{lstlisting}

\subsubsection{Server}
\label{sec:overview:runningexample:server}

We implement the server computation $\server$ as the following function:
\begin{lstlisting}
let server batchSize =
    let rec waitForBatchSize () =
        promise (batchSizeRequest () |->
            send batchSizeResponse batchSize;
            waitForBatchSize ()
        ) as p in return p
    in
    let rec waitForRequest () =
        promise (request offset |->
            let payload = map (fun x |-> 10 * x) (range offset (offset + batchSize - 1)) in
            send response payload;
            waitForRequest ()
        ) as p in return p
    in
    waitForBatchSize (); waitForRequest ()
\end{lstlisting}
where the computation \lstinline{range i j} returns a list of integers ranging from \lstinline{i} to \lstinline{j} (both inclusive).

The server simply installs two recursively defined interrupt handlers: the first 
one listens for and responds to client's requests about the batch size it uses; 
and the second one responds to client's requests for new data. Both interrupt 
handlers then simply recursively reinstall themselves.

\subsubsection{User}
\label{sec:overview:runningexample:user}

We can also simulate the user as a computation. Namely, we implement 
it as a function that every now and then issues a request to 
the client to display the next data item:

\begin{lstlisting}
let rec user () =
    let rec wait n = 
        if n = 0 then return () else wait (n - 1)
    in
    send nextItem (); wait 10; user ()
\end{lstlisting}
It is straightforward to extend the user also with a handler for $\opsym{display}$ interrupts (we omit it here).

\subsubsection{Running the Server, Client, and User in Parallel}
\label{sec:overview:runningexample:parallel}

Finally, we can simulate our running example in full by running all 
three computations we defined above as parallel processes, e.g., as follows: 
\begin{lstlisting}
run (server 42) || run (client ()) || run (user ())
\end{lstlisting}

% !TEX root = paper.tex

\section{A Calculus for Asynchronous Effects: Values and Computations}
\label{sec:basic-calculus:computations}

We now distil the ideas we introduced in the previous section into a core 
calculus for programming with asynchronous effects, called \lambdaAEff. It 
is based on \citeauthor{Levy:FGCBV}'s [\citeyear{Levy:FGCBV}] fine-grain 
call-by-value $\lambda$-calculus (FGCBV), and as such, it is a low-level 
intermediate language to which a corresponding high-level user-facing programming language 
could be compiled to.
In order to better explain the different features of the calculus and its semantics, we split 
\lambdaAEff~into a \emph{sequential} part (discussed below) and a 
\emph{parallel} part (discussed in \autoref{sec:basic-calculus:processes}). 
We note that this separation is purely presentational.

\subsection{Values and Computations}
\label{sec:basic-calculus:values-and-computations}

The syntax of terms is given in \autoref{fig:terms}, stratified into \emph{values} and \emph{computations}, 
as in FGCBV.

\begin{figure}[tp]
  \parbox{\textwidth}{
  \centering
  \small
  \begin{align*}
  \intertext{\textbf{Values}}
  V, W
  \bnfis& x                                       & &\text{variable} \\
  \bnfor& \tmunit \bnfor\! \tmpair{V}{W}                                & &\text{unit and pair} \\
  \bnfor& \tminl[Y]{V} \bnfor\! \tminr[X]{V}    & &\text{left and right injections} \\
  \bnfor& \tmfun{x : X}{M}                        & &\text{function abstraction} \\
  \bnfor& \tmpromise V                            & &\text{fulfilled promise}
  \\[1ex]
  \intertext{\textbf{Computations}}
  M, N
  \bnfis& \tmreturn{V}                            & &\text{returning a value} \\
  \bnfor& \tmlet{x}{M}{N}          & &\text{sequencing} \\
  \bnfor& \tmletrec[: \tyfun{X}{\tycomp{Y}{(\o,\i)}}]{f}{x}{M}{N} & &\text{recursive definition} \\
  \bnfor& V\,W                                    & &\text{function application} \\
  \bnfor& \tmmatch{V}{\tmpair{x}{y} \mapsto M}    & &\text{product elimination} \\
  \bnfor& \tmmatch[\tycomp{Z}{(\o,\i)}]{V}{}                        & &\text{empty elimination} \\
  \bnfor& \tmmatch{V}{\tminl{x} \mapsto M, \tminr{y} \mapsto N}
                                                  & &\text{sum elimination} \\
  \bnfor& \tmopout{op}{V}{M}       & &\text{outgoing signal} \\
  \bnfor& \tmopin{op}{V}{M}          & &\text{incoming interrupt} \\
  \bnfor& \tmwith{op}{x}{M}{p}{N}      & &\text{interrupt handler} \\
  \bnfor& \tmawait{V}{x}{M}             & &\text{awaiting a promise to be fulfilled}
  \end{align*}
  } 
  \caption{Values and computations.}
  \label{fig:terms}
\end{figure}

\paragraph{Values}

The values $V,W,\ldots$ are mostly standard. They include  
variables, introduction forms for 
sums and products, and functions. The only \lambdaAEff-specific value 
is $\tmpromise V$, which denotes a \emph{fulfilled promise}, indicating that the promise of 
handling some interrupt has been fulfilled with the value $V$. 

\paragraph{Computations} The computations $M,N,\ldots$ also include all 
standard terms from FGCBV: 
returning values, sequencing, recursion, function 
application, and elimination forms.
Note that we annotate recursive definitions with the type of the function being defined, 
including the annotations $(\o,\i)$ that describe the possible effects of the function.
While we do not study effect inference in this paper, our experience 
is that these annotations should make it possible 
to fully infer types.

The first two computations specific to \lambdaAEff~are \emph{signals} $\tmopout{op}{V}{M}$ and 
\emph{interrupts} $\tmopin{op}{V}{M}$, where the name 
$\opsym{op}$ is drawn from an assumed set $\sig$ of names, $V$ is a data  
payload, and $M$ is a continuation.

The next \lambdaAEff-specific computation is the \emph{interrupt handler} $\tmwith{op}{x}{M}{p}{N}$, 
where $x$ is bound in $M$ and $p$ in $N$.
As discussed in the previous section, one should understand this computation as making a promise 
to handle a future incoming interrupt $\opsym{op}$ by executing the computation $M$. Sub-computations of the continuation 
$N$ can then explicitly await, when necessary, this promise to be fulfilled by blocking on the \emph{promise variable} $p$
using the final \lambdaAEff-specific computation term, the \emph{awaiting} construct $\tmawait{V}{x}{M}$.
We note that $p$ is an ordinary variable---it just gets assigned the distinguished promise type by the interrupt handler.

\subsection{Small-Step Operational Semantics}
\label{sec:basic-calculus:semantics:computations}

We equip \lambdaAEff~with an evaluation contexts based 
small-step operational semantics, 
defined using a reduction relation $M \reduces N$. 
The \emph{reduction rules} and \emph{evaluation contexts} are given 
in \autoref{fig:small-step-semantics-of-computations}.

\begin{figure}[tp]
  \small
  \begin{align*}
    \intertext{\textbf{Standard computation rules}}
    \tmapp{(\tmfun{x \of X}{M})}{V} &\reduces M[V/x]
    \\
    \tmlet{x}{(\tmreturn V)}{N} &\reduces N[V/x]
    \\
    \tmmatch{\tmpair{V}{W}}{\tmpair{x}{y} \mapsto M} &\reduces M[V/x, W/y]
    \\
    \mathllap{
      \tmmatch{(\tminl[Y]{V})}{\tminl{x} \mapsto M, \tminr{y} \mapsto N} 
    } &\reduces
    M[V/x]
    \\
    \mathllap{
      \tmmatch{(\tminr[X]{W})}{\tminl{x} \mapsto M, \tminr{y} \mapsto N}
    } &\reduces
    N[W/y]
    \\
    \tmletrec[: \tyfun{X\!}{\tycomp{\!Y\!}{\!(\o,\i)}}]{f}{x}{M}{N}  &\reduces
      N[\tmfun{x \of X}{\tmletrec[: \tyfun{X\!}{\tycomp{\!Y\!}{\!(\o,\i)}}]{f}{x}{M}{M}}/f]
    \\[1ex]
    \intertext{\textbf{Algebraicity of signals and interrupt handlers}}
    \tmlet{x}{(\tmopout{op}{V}{M})}{N} &\reduces \tmopout{op}{V}{\tmlet{x}{M}{N}}
    \\
    \tmlet{x}{(\tmwith{op}{y}{M}{p}{N_1})}{N_2} &\reduces \tmwith{op}{y}{M}{p}{(\tmlet{x}{N_1}{N_2})}
    \\[1ex]
    \intertext{\textbf{Commutativity of signals with interrupt handlers}}
    \tmwith{op}{x}{M}{p}{\tmopout{op'}{V}{N}} &\reduces \tmopout{op'}{V}{\tmwith{op}{x}{M}{p}{N}}
    \\[1ex]
    \intertext{\textbf{Interrupt propagation}}
    \tmopin{op}{V}{\tmreturn W} &\reduces \tmreturn W
    \\
    \tmopin{op}{V}{\tmopout{op'}{W}{M}} &\reduces \tmopout{op'}{W}{\tmopin{op}{V}{M}}
    \\
    \tmopin{op}{V}{\tmwith{op}{x}{M}{p}{N}} &\reduces \tmlet{p}{M[V/x]}{\tmopin{op}{V}{N}}
    \\
    \tmopin{op'}{V}{\tmwith{op}{x}{M}{p}{N}} &\reduces \tmwith{op}{x}{M}{p}{\tmopin{op'}{V}{N}}
    \quad {\color{rulenameColor}(\op \neq \op')}
    \\[-6ex]
    \end{align*}
  \begin{minipage}[t]{0.4\textwidth}
    \begin{align*}
    \intertext{\quad\,\textbf{Awaiting a promise to be fulfilled}}
    \tmawait{\tmpromise V}{x}{M} \reduces M[V/x]
    \end{align*}
  \end{minipage}
  \begin{minipage}[t]{0.4\textwidth}
    \centering
    \begin{align*}
    \intertext{\textbf{Evaluation context rule}}
    \coopinfer{}{
      M \reduces N
    }{
      \E[M] \reduces \E[N]
    }
    \end{align*}
  \vspace{-4ex}
  \end{minipage}
  \begin{align*}
  \intertext{\textbf{where}\vspace{1ex}}
  \text{$\E$}
  \bnfis [~]
  \bnfor \tmlet{x}{\E}{N}
  \bnfor \tmopout{op}{V}{\E}
  \bnfor \tmopin{op}{V}{\E} 
  \bnfor \tmwith{op}{x}{M}{p}{\E}
  \end{align*}
  \caption{Small-step operational semantics of computations.}
  \label{fig:small-step-semantics-of-computations}
\end{figure}

\paragraph{Computation rules}
The first group includes \emph{standard reduction rules} from FGCBV, such as $\beta$-reducing function applications, sequential composition, and the standard elimination forms. The semantics also includes a rule for unfolding general-recursive definitions. 
These rules involve standard \emph{capture avoiding substitutions} $M[V/x]$, defined by straightforward structural recursion.

\paragraph{Algebraicity}
This group of reduction rules \emph{propagates outwards} the signals (resp.~interrupt handlers) that have been issued (resp.~installed) in 
sub-computations. While it is not surprising that outgoing signals 
behave like algebraic \emph{operation calls}, getting propagated outwards as far as possible, then it is much more curious that 
the natural operational behaviour of interrupt handlers turns out to be the same. As we shall explain in \autoref{sec:conclusion},  
despite using the (systems-inspired) ``handler'' terminology, mathematically interrupt handlers are in fact a form of algebraic operations.

\paragraph{Commutativity of signals with interrupt handlers}
This rule complements the algebraicity rule for signals, by further propagating 
them outwards, past any enveloping interrupt handlers. From the perspective of algebraic effects, 
this rule is an example of two algebraic operations \emph{commuting}.
For this rule to be type safe, the type system ensures that the (promise) variable $p$ cannot appear in $V$. 

\paragraph{Interrupt propagation}
The handler-operation curiosity does not end with interrupt handlers. This group of reduction rules describes how  
interrupts are \emph{propagated inwards} into sub-computations. While $\tmopin{op}{V}{M}$ might look like a conventional  
operation call, then its operational behaviour instead mirrors that of (deep) \emph{effect handling}, where one also recursively descends into
the computation being handled. The first reduction rule states that we can safely discard an interrupt when it reaches a terminal 
computation $\tmreturn W$. The second rule states that we can propagate incoming interrupts past any outward moving signals. The last 
two rules describe how interrupts interact with interrupt handlers, in particular, that the former behave like effect handling 
(when understanding interrupt handlers as generalised algebraic operations). On the one hand, if the interrupt 
matches the interrupt handler it encounters, the corresponding handler code $M$ is executed, and the interrupt is 
propagated inwards into the continuation $N$. On the other hand, if the interrupt 
does not match the interrupt handler, it is simply propagated past the interrupt handler into $N$.

\paragraph{Awaiting}
The semantics includes a $\beta$-rule for the $\tmkw{await}$ construct, allowing the blocked computation $M$
to proceed executing as $M[V/x]$ when $\tmkw{await}$ is given a fulfilled promise $\tmpromise V$. 

\paragraph{Evaluation contexts}
The semantics allows reductions under \emph{evaluation contexts} $\E$.
Observe that  as discussed earlier, the inclusion of interrupt handlers in the evaluation contexts means that reductions
involve potentially open terms. 
Also, differently from the semantics of conventional operation calls \cite{Kammar:Handlers,Bauer:EffectSystem}, 
our evaluation contexts include outgoing signals. As such, the \emph{evaluation context rule} allows the execution of a computation 
to proceed even if a signal has not yet been propagated to its receiver, or when an interrupt has 
not yet arrived. Importantly, the evaluation contexts do not include $\tmkw{await}$, so as to model its intended blocking behaviour.
We write $\E[M]$ for the recursive operation of filling the hole $[~]$ in $\E$ with $M$.

\paragraph{Non-confluence}
It is worth noting that the asynchronous design means that the operational semantics 
is naturally \emph{nondeterministic}. But more interestingly, the semantics is also \emph{not confluent}.

For one source of non-confluence, let us consider two reduction sequences of a same (closed) computation, 
where for better readability, we highlight the active redex for each reduction step:
 \[
\hspace{-0.15cm}
\begin{array}{r@{\,} l}
  & \tmopin{op}{V}{\tmwith{op}{x}{(\tmwith{op'}{y}{M}{q}{\tmawait{q}{z}{M'}})}{p}{\!\highlightgray{N}}} 
  \\[1ex]
  \reduces & \highlightgray{\tmopin{op}{V}{\tmwith{op}{x}{(\tmwith{op'}{y}{M}{q}{\tmawait{q}{z}{M'}})}{p}{\!\highlightwhite{N'}}}}
  \\[1ex]
  \reduces & \highlightgray{\tmlet{p}{(\tmwith{op'}{y}{M[V/x]}{q}{\tmawait{q}{z}{M'}})}{\!\highlightwhite{\tmopin{op}{V}{N'}}}}
  \\[1ex]
  \reduces & \tmwith{op'}{y}{M[V/x]}{q}{\tmawait{q}{z}{(\tmlet{p}{M'}{\tmopin{op}{V}{N'}})}}
\end{array}
\]
and
\[
\hspace{-0.15cm}
\begin{array}{r@{\,} l}
  & \highlightgray{\tmopin{op}{V}{\tmwith{op}{x}{(\tmwith{op'}{y}{M}{q}{\tmawait{q}{z}{M'}})}{p}{\!\highlightwhite{N}}}}
  \\[1ex]
  \reduces & \highlightgray{\tmlet{p}{(\tmwith{op'}{y}{M[V/x]}{q}{\tmawait{q}{z}{M'}})}{\!\highlightwhite{\tmopin{op}{V}{N}}}}
  \\[1ex]
  \reduces & \tmwith{op'}{y}{M[V/x]}{q}{\tmawait{q}{z}{(\tmlet{p}{M'}{\tmopin{op}{V}{N}})}}
\end{array}
\]
Here, both final computations are \emph{temporarily} blocked until an incoming interrupt $\opsym{op'}$
is propagated to them and the (promise) variable $q$ gets bound to a fulfilled promise. Until this happens, 
it is not possible for the blocked continuation $N$ to reduce to $N'$ in the latter final computation.

Another distinct source of non-confluence concerns the commutativity of outgoing signals with enveloping interrupt 
handlers. For instance, the following (closed) composite computation
\[
\tmopin{op}{V}{{\tmwith {op} x {\tmopout{op'}{W'}{M}} p {\tmopout{op''}{W''}{N}}}}
\]
can nondeterministically reduce to either of the next two computations:
\[
\tmopout{op'}{W'}{\tmopout{op''}{W''}{{\tmlet{p}{M}{\tmopin{op}{V}{N}}}}}
\quad
\tmopout{op''}{W''}{\tmopout{op'}{W'}{{\tmlet{p}{M}{\tmopin{op}{V}{N}}}}}
\]
depending on whether we first propagate the interrupt $\op$ inwards or 
the signal $\op''$ outwards. As a result, in the resulting two computations, 
the signals $\op'$ and $\op''$ get issued in a different order.

\subsection{Type-and-Effect System}
\label{sec:basic-calculus:type-system:computations}

We equip \lambdaAEff~with a type system in the tradition of type-and-effect systems for algebraic effects and 
effect handlers \cite{Bauer:EffectSystem,Kammar:Handlers}, by extending the simple type system of FGCBV 
with annotations about possible effects in function and computation types.

\subsubsection{Types}
\label{sec:basic-calculus:type-system:computations:types}

We define types in \autoref{fig:types}, separated into ground, value, and computation types.

As noted in \autoref{sec:basic-calculus:values-and-computations}, \lambdaAEff~is parameterised over a set 
$\sig$ of signal and interrupt \emph{names}. 
To each such name $\op \in \sig$, we assign a fixed 
\emph{signature} $\op : A_\op$ that specifies the type $A_\op$ of the payload of the corresponding 
signal or interrupt.
Crucially, in order to be able to later prove that \lambdaAEff~is type safe 
(see \autoref{theorem:progress}, but also the relevant discussion in \autoref{sec:conclusion}), 
we restrict these signatures to \emph{ground types} $A,B,\ldots$, 
which include standard base, unit, empty, product, and sum types.

\begin{figure}[tb]
  \parbox{\textwidth}{
  \centering
  \small
  \begin{align*}
  \text{Ground type $A$, $B$}
  \bnfis& \tybase \,\bnfor\! \tyunit \,\bnfor\! \tyempty \,\bnfor\! \typrod{A}{B} \,\bnfor\! \tysum{A}{B}
  \\[1ex]
  \text{Signal or interrupt signature:}
  \phantom{\bnfis}& \op : A_\op
  \\[1ex]
  \text{Outgoing signal annotations:}
  \phantom{\bnfis}& \o \in O 
  \\
  \text{Interrupt handler annotations:}
  \phantom{\bnfis}& \i \in I 
  \\[1ex]
  \text{Value type $X$, $Y$}
  \bnfis& A \,\bnfor\! \typrod{X}{Y} \,\bnfor\! \tysum{X}{Y} \,\bnfor\! \tyfun{X}{\tycomp{Y}{(\o,\i)}} \,\bnfor\! \typromise{X}
  \\
  \text{Computation type:}
  \phantom{\bnfis}& \tycomp{X}{(\o,\i)}
  \end{align*}
  } 
  \caption{Value and computation types}
  \label{fig:types}
\end{figure}

\emph{Value types} $X,Y,\ldots$ extend ground types with function and promise types.
The \emph{function type} $\tyfun{X}{\tycomp{Y}{(\o,\i)}}$ classifies functions that take $X$-typed arguments 
to computations classified by the \emph{computation type} $\tycomp{Y}{(\o,\i)}$, i.e., ones that return $Y$-typed 
values, while possibly issuing signals specified by $\o$ and handling interrupts specified by $\i$. 
The \emph{effect annotations} $\o$ and $\i$ are drawn from sets $O$ and $I$ whose definitions we discuss  
below in \autoref{sec:basic-calculus:effect-annotations}. Finally, the \lambdaAEff-specific \emph{promise type} 
$\typromise{X}$ classifies those promises that can be fulfilled by supplying a value of type $X$.

\subsubsection{Effect Annotations}
\label{sec:basic-calculus:effect-annotations}

We now explain how we define the sets $O$ and $I$ from which we draw the 
effect annotations that we use for specifying functions and computations.
Traditionally, effect systems for algebraic effects simply use (flat) sets of 
operation names for effect annotations \cite{Bauer:EffectSystem,Kammar:Handlers}. 
In \lambdaAEff, however, we need to be 
more careful, because triggering an interrupt handler executes a computation 
that can issue potentially different signals and handle different interrupts from the main 
program, and we would like to capture this in types.
 
\paragraph{Signal annotations}
First, as outgoing signals do not carry any computational data, we follow
the tradition of type-and-effect systems for algebraic effects, and let 
$O$ be the \emph{power set} $\Pow \sig$. As such, each $\o \in O$ is a subset of 
the signature $\Sigma$, specifying which signals a computation might issue.

\paragraph{Interrupt handler annotations}
As noted above, for specifying installed interrupt handlers, we cannot use (flat) sets 
of interrupt names as the effect annotations $\i \in I$ if we want to track the nested effectful 
structure. Instead, 
we define $I$ as the \emph{greatest fixed point}
of a set functor $\Phi$ given by
\[
\Phi (X) \defeq \sig \Rightarrow (O \times X)_\bot
\]
where $\Rightarrow$ is exponentiation, $\times$ is Cartesian product, 
and $(-)_\bot$ is the lifting operation $X_\bot \defeq X \cupdot \{\bot\}$, and 
where $\cupdot$ is the disjoint union of sets. Formally speaking, $I$ is given 
by an isomorphism $I \cong \Phi(I)$, but for presentation purposes we leave it 
implicit and work as if we had a strict equality $I = \Phi(I)$.

Intuitively, each $\i \in I$ is a  
\emph{possibly infinite nesting of partial mappings} of pairs of $O$- and $I$-annotations to names in 
$\sig$---these pairs of annotations classify the possible effects of the corresponding interrupt handler code.
We use the 
record notation $\i = \{ \op_1 \mapsto (\o_1,\i_1) , \ldots , \op_n \mapsto (\o_n,\i_n) \}$ to mean that $\i$ maps the names $\op_1, \ldots, \op_n$ to the annotations $(\o_1,\i_1), \ldots, (\o_n,\i_n)$, and any other names in $\sig$ to $\bot$.
We write $\i\, (\op_i) = (\o_i,\i_i)$ to mean that the annotation $\i$ maps $\op_i$ to $(\o_i,\i_i)$.

\paragraph{Subtyping and recursive effect annotations}
Both $O$ and $I$ come equipped with natural \emph{partial orders}: for $O$, $\order O$ is given simply by 
subset inclusion; and for $I$, $\order I$ is characterised as follows:
\[
\begin{array}{l c l}
\i \order I \i'
&
\text{iff}
&
\forall\, (\op \in \sig) \, (\o'' \in O) \, (\i'' \in I) .\, \i\, (\op) = ({\o''} , {\i''}) \implies 
\\[0.5ex]
&& \exists\, (\o''' \in O) \, (\i''' \in I) .\, \i'\, (\op) = ({\o'''} , {\i'''}) \wedge \o'' \order O \o''' \wedge \i'' \order I \i'''
\end{array}
\]
We often also use the \emph{product order} $\order {O \times I}$, defined as
$(\o,\i) \order {O \times I} (\o',\i') \defeq \o \order O \o' \wedge \i \order I \i'$.
In particular, we use $\order {O \times I}$ in \autoref{sect:typing-rules} to define the subtyping 
relation for \lambdaAEff's computation types.

Importantly, the partial orders $(O,\order O)$ and $(I,\order I)$ are both \emph{$\omega$-complete} and \emph{pointed}, i.e., 
they are \emph{pointed cpos}, meaning they have least upper bounds of all increasing $\omega$-chains, and 
least elements (given by the empty set $\emptyset$ and the constant $\bot$-valued mapping, respectively). 
As a result, \emph{least fixed points} of continuous (endo)maps on them are guaranteed to exist.
We refer the interested reader to \citet{Amadio:Domains} and \citet{Gierz:ContinuousLattices} for additional 
domain-theoretic background.

For \lambdaAEff, we are particularly interested in the least fixed points of continuous maps $f : I \to I$, 
so as to specify and typecheck recursive interrupt handler examples, as we illustrate in  
\autoref{sec:basic-calculus:rec-handler-typing}. 

We also note that if we were only interested in the type safety of \lambdaAEff, and not 
in typechecking recursively defined interrupt handlers, then we would not need $(I,\order I)$ to be \emph{$\omega$-complete}, 
and could have instead chosen $I$ to be the 
\emph{least fixed point} of $\Phi$, which is what we do in our \pl{Agda} 
formalisation. In this case, each interrupt handler annotation $\i \in I$ would be a \emph{finite nesting of partial mappings}. 

Finally, we envisage that any future full-fledged high-level language based on \lambdaAEff~would 
allow users to define their (recursive) effect annotations in a small domain-specific language, providing  
a syntactic counterpart to the domain-theoretic development we use for typing \lambdaAEff~in this paper.

\subsubsection{Typing Rules}
\label{sect:typing-rules}

We characterise \emph{well-typed values} using the judgement $\Gamma \types V : X$ 
and \emph{well-typed computations} using $\Gamma \types M : \tycomp{X}{(\o,\i)}$.
In both judgements, $\Gamma$ is a \emph{typing context} of the form $x_1 \of X_1, \ldots, x_n \of X_n$.
The rules defining these judgements are respectively given in \autoref{fig:value-typing-rules} and 
\ref{fig:computation-typing-rules}.

\begin{figure}[tp]
  \centering
  \small
  \begin{mathpar}
  \coopinfer{TyVal-Var}{
  }{
    \Gamma, x \of X, \Gamma' \types x : X
  }
  \qquad
  \coopinfer{TyVal-Unit}{
  }{
    \Gamma \types \tmunit : \tyunit
  }
  \qquad
  \coopinfer{TyVal-Pair}{
    \Gamma \types V : X \\
    \Gamma \types W : Y
  }{
    \Gamma \types \tmpair{V}{W} : \typrod{X}{Y}
  }
  \qquad
  \coopinfer{TyVal-Promise}{
    \Gamma \types V : X
  }{
    \Gamma \types \tmpromise V : \typromise X
  }
  \\
  \coopinfer{TyVal-Inl}{
    \Gamma \types V : X
  }{
    \Gamma \types \tminl[Y]{V} : X + Y
  }
  \qquad
  \coopinfer{TyVal-Inr}{
    \Gamma \types W : Y
  }{
    \Gamma \types \tminr[X]{W} : X + Y
  }
  \qquad
  \coopinfer{TyVal-Fun}{
    \Gamma, x \of X \types M : \tycomp{Y}{(\o,\i)}
  }{
    \Gamma \types \tmfun{x : X}{M} : \tyfun{X}{\tycomp{Y}{(\o,\i)}}
  }
  \end{mathpar}
  \caption{Value typing rules.}
  \label{fig:value-typing-rules}
\end{figure}

\begin{figure}[tp]
  \centering
  \small
  \begin{mathpar}
  \coopinfer{TyComp-Return}{
    \Gamma \types V : X
  }{
    \Gamma \types \tmreturn{V} : \tycomp{X}{(\o,\i)}
  }
  \qquad
  \coopinfer{TyComp-Let}{
    \Gamma \types M : \tycomp{X}{(\o,\i)}
    \\
    \Gamma, x \of X \types N : \tycomp{Y}{(\o,\i)}
  }{
    \Gamma \types
    \tmlet{x}{M}{N} : \tycomp{Y}{(\o,\i)}
  }
  \\
  \coopinfer{TyComp-LetRec}{
    \Gamma, f \of \tyfun{X}{\tycomp{Y}{(\o,\i)}}, x \of X \types M : \tycomp{Y}{(\o,\i)}
    \\
    \Gamma, f \of \tyfun{X}{\tycomp{Y}{(\o,\i)}} \types N : \tycomp{Z}{(\o',\i')}
  }{
    \Gamma \types
    \tmletrec[: \tyfun{X}{\tycomp{Y}{(\o,\i)}}]{f}{x}{M}{N} : \tycomp{Z}{(\o',\i')}
  }
  \\
  \coopinfer{TyComp-Apply}{
    \Gamma \types V : \tyfun{X}{\tycomp{Y}{(\o,\i)}} \\
    \Gamma \types W : X
  }{
    \Gamma \types \tmapp{V}{W} : \tycomp{Y}{(\o,\i)}
  }
  \qquad
  \coopinfer{TyComp-MatchPair}{
    \Gamma \types V : \typrod{X}{Y} \\
    \Gamma, x \of X, y \of Y \types M : \tycomp{Z}{(\o,\i)}
  }{
    \Gamma \types \tmmatch{V}{\tmpair{x}{y} \mapsto M} : \tycomp{Z}{(\o,\i)}
  }
  \\
  \coopinfer{TyComp-MatchEmpty}{
    \Gamma \types V : \tyempty
  }{
    \Gamma \types \tmmatch[\tycomp{Z}{(\o,\i)}]{V}{} : \tycomp{Z}{(\o,\i)}
  }
  \qquad
  \coopinfer{TyComp-MatchSum}{
    \Gamma \types V : X + Y \\\\
    \Gamma, x \of X \types M : \tycomp{Z}{(\o,\i)} \\
    \Gamma, y \of Y \types N : \tycomp{Z}{(\o,\i)} \\
  }{
    \Gamma \types \tmmatch{V}{\tminl{x} \mapsto M, \tminr{y} \mapsto N} : \tycomp{Z}{(\o,\i)}
  }
  \\
  \coopinfer{TyComp-Signal}{
    \op \in \o \\
    \Gamma \types V : A_\op \\
    \Gamma \types M : \tycomp{X}{(\o,\i)} 
  }{
    \Gamma \types \tmopout{op}{V}{M} : \tycomp{X}{(\o,\i)}
  }
  \qquad
  \coopinfer{TyComp-Interrupt}{
    \Gamma \types V : A_\op \\
    \Gamma \types M : \tycomp{X}{(\o,\i)} 
  }{
    \Gamma \types \tmopin{op}{V}{M} : \tycomp{X}{\opincomp {op} (\o,\i)}
  }
  \\
  \coopinfer{TyComp-Promise}{
    \i\, (\op) = ({\o'} , {\i'}) \\
    \Gamma, x \of A_\op \types M : \tycomp{\typromise X}{(\o',\i')} \\
    \Gamma, p \of \typromise X \types N : \tycomp{Y}{(\o,\i)} 
  }{
    \Gamma \types \tmwith{op}{x}{M}{p}{N} : \tycomp{Y}{(\o,\i)}
  }
  \\
  \coopinfer{TyComp-Await}{
    \Gamma \types V : \typromise X \\
    \Gamma, x \of X \types M : \tycomp{Y}{(\o,\i)} 
  }{
    \Gamma \types \tmawait{V}{x}{M} : \tycomp{Y}{(\o,\i)}
  }
  \qquad
   \coopinfer{TyComp-Subsume}{
      \Gamma \types M : \tycomp{X}{(\o, \i)} \\
      (\o,\i) \order {O \times I} (\o',\i')
    }{
      \Gamma \types M : \tycomp{X}{(\o', \i')}
    }
  \end{mathpar}
  \caption{Computation typing rules.}
  \label{fig:computation-typing-rules}
\end{figure}

\paragraph{Values}

The rules for values are mostly standard.
The only \lambdaAEff-specific rule is \textsc{TyVal-Promise}, which states that in order to fulfil 
a \emph{promise} of type $\typromise X$, one has to supply a value of type $X$.

\paragraph{Computations}

Analogously to values, the typing rules are standard for the computation terms that \lambdaAEff~inherits from FGCBV, 
with the \lambdaAEff-rules additionally tracking the effect information $(\o,\i)$.

The \lambdaAEff-specific rule \textsc{TyComp-Signal} states that in order 
to issue a signal $\op$ in a computation with type $\tycomp{X}{(\o,\i)}$, we must have $\op \in \o$ and the type of the payload
has to match $\op$'s signature. 

The rule \textsc{TyComp-Promise} states that 
the interrupt handler code $M$ has to return a fulfilled promise of type $\typromise X$, for some type $X$, 
while possibly issuing signals $\o'$ and handling interrupts $\i'$, both of which are 
determined by the effect annotation $\i$ of the entire computation, i.e., 
$\i\, (\op) = (\o',\i')$. The variable $p$ bound in the continuation, which sub-computations can block on
to await a given interrupt to arrive and be handled, also gets assigned 
the promise type $\typromise X$. It is worth noting that one could have had $M$ simply 
return values of type $X$, but at the cost of not being able to implement some of the more interesting examples, 
e.g., guarded interrupt handlers in \autoref{sec:applications:guarder-handlers}.
At the same time, for \lambdaAEff's type safety, it is 
crucial that $p$ would have remained assigned the promise type $\typromise X$.

The rule \textsc{TyComp-Await} simply says that when awaiting a promise of type $\typromise X$ to be fulfilled, 
the continuation $M$ can refer to the promised value (in the future) using the variable $x$ of type $X$.

The rule \textsc{TyComp-Interrupt} is used to type incoming interrupts. 
In particular, when the outside world propagates an interrupt $\op$ to a computation 
$M$ of type $\tycomp{X}{(\o,\i)}$, the resulting 
computation $\tmopin{op}{V}{M}$ gets assigned the type $\tycomp{X}{\opincomp {op} (\o,\i)}$, 
where the interrupt $\op$ also \emph{acts} on the effect annotations. Intuitively, 
$\opincomp {op} (\o,\i)$ mimics the act of triggering interrupt
handlers for $\op$ at the level of effect annotations. Formally, we define this 
\emph{action of interrupts} on effect annotations as follows:
\[
\opincomp {op} {(\o , \i)}
~\defeq~
  \begin{cases}
   \left(\o \cup \o' , \i[\op \mapsto \bot] \cup \i' \right) & \mbox{if } \i\, (\op) = (\o',\i')\\
   (\o,\i) & \mbox{otherwise} 
  \end{cases}
\vspace{-0.25ex} % To bump the code snippet below to this page
\]
In other words, if $M$ has any interrupt handlers installed for $\op$, then $\i\, (\op) = (\o',\i')$, 
where $(\o',\i')$ specifies the effects of said interrupt handler code. Now, when the inward propagating 
interrupt $\op$ reaches those interrupt handlers, it triggers the execution of the corresponding handler code, 
and thus the entire computation $\tmopin{op}{V}{M}$ can also issue signals in $\o'$ and handle interrupts in $\i'$.

The notation $\i[\op \mapsto \bot]$ sets $\i$ to $\bot$ at $\op$, 
and leaves it unchanged elsewhere.
In particular, mapping $\op$ to $\bot$ captures that the interrupt triggers all corresponding interrupt handlers in $M$.

The \emph{join-semilattice} structure 
$\o \cup \o' \in O$ is given by the union of sets, while 
$\i \cup \i' \in I$ is given by
\[
\i \cup \i'
~\defeq~
\lam {\op} 
\begin{cases}
(\o'' \cup \o''' , \i'' \cup \i''') & \mbox{if } \i\, (\op) = (\o'',\i'') \wedge \i'\, (\op) = (\o''',\i''') \\
(\o'' , \i'') & \mbox{if } \i\, (\op) = (\o'',\i'') \wedge \i'\, (\op) = \bot \\
(\o''' , \i''') & \mbox{if } \i\, (\op) = \bot \wedge \i'\, (\op) = (\o''',\i''') \\
\bot & \mbox{if } \i\, (\op) = \bot \wedge \i'\, (\op) = \bot \\
\end{cases}
\vspace{-0.25ex} % To bump the code snippet below to this page
\]

We also note that the action $\opincomp {op} {(-)}$ has various useful properties, which we use in later proofs
(where we write $\pi_1$ and $\pi_2$ for the two projections associated with the Cartesian product $O \times I$):

\begin{lemma}
\label{lemma:action}
\mbox{}
\begin{enumerate}
\item $\o \order O \pi_1\, (\opincomp {op} {(\o,\i)})$.
\item If $\i\, (\op) = (\o',\i')$, then $(\o',\i') \order {O \times I} \opincomp {op} {(\o,\i)}$.
\item If $\op \neq \op'$ and $\i\, (\op') = (\o',\i')$, then $(\o',\i') \order {O \times I} (\pi_2\, (\opincomp {op} {(\o,\i)}))\, (\op')$.
\end{enumerate}
\end{lemma}

Finally, the rule \textsc{TyComp-Subsume} allows \emph{subtyping}. 
To simplify the presentation, we consider a limited form of subtyping, in which we 
shallowly relate only signal and interrupt annotations.

\subsubsection{Typechecking Recursively Defined Interrupt Handlers}
\label{sec:basic-calculus:rec-handler-typing}

We conclude discussing \lambdaAEff's type-and-effect system by briefly returning to the reason 
why we defined our effect annotations using lightweight domain theory in the first place, namely, 
so as to typecheck recursive interrupt handlers.

As an example, we recall the following fragment of the server code from \autoref{sec:overview:runningexample:server}:
\begin{lstlisting}
let rec waitForBatchSize () =
    promise (batchSizeReq () |-> send batchSizeResp batchSize; waitForBatchSize ()) as p in return p
\end{lstlisting}
Here, \ls$waitForBatchSize ()$ is an interrupt handler for $\opsym{batchSizeReq}$ 
that recursively reinstalls itself immediately 
after issuing a $\opsym{batchSizeResp}$ signal. Due to its recursive definition, 
it is not surprising that the type of \ls$waitForBatchSize$ should also be given recursively, in 
particular, if we want to give it a more precise type than one which simply says that any effect is possible.

To this end, we assign \ls$waitForBatchSize$ the type 
$\tyfun{\tyunit}{\tycomp{\typromise \tyunit}{(\emptyset, \i_{\text{b}})}}$, where 
$\i_{\text{b}}$ is the \emph{least fixed point} of the continuous map 
$\i \mapsto \{~ \opsym{batchSizeReq} \mapsto (\{\opsym{batchSizeResp}\} , \i) ~\} : I \to I$, i.e.,  
\[
\i_{\text{b}} = \big\{~ \opsym{batchSizeReq} \mapsto (\{\opsym{batchSizeResp}\} , \{~ \opsym{batchSizeReq} \mapsto (\{\opsym{batchSizeResp}\} , ~\ldots~) ~\}) ~\big\}
\]
As such, $(\emptyset, \i_{\text{b}})$  
captures that at the top level \ls$waitForBatchSize ()$ 
installs an interrupt handler and issues no signals, and that every $\opsym{batchSizeReq}$
interrupt causes a signal to be issued and the interrupt handler 
to be reinstalled.
Checking that \ls$waitForBatchSize$ has the 
type $\tyfun{\tyunit}{\tycomp{\typromise \tyunit}{(\emptyset, \i_{\text{b}})}}$  
involves unfolding the definition of $\i_{\text{b}}$ and using 
subtyping. The latter is needed when we recursively call \ls$waitForBatchSize ()$
where a computation of type $\tycomp{\typromise \tyunit}{(\{\opsym{batchSizeResp}\}, \i_{\text{b}})}$ 
is expected.

\subsection{Type Safety}
\label{sec:basic-calculus:type-safety}

We now prove type safety for the sequential part of \lambdaAEff, 
showing that ``well-typed programs do not go wrong''. 
As usual, we split type 
safety into \emph{progress} and \emph{preservation} \cite{Wright:SynAppTypeSoundness}.

\subsubsection{Progress}
\label{sec:basic-calculus:type-safety:progress}

The progress result says that well-typed closed computations can either make a step of  
reduction, or are already in a well-defined result form (and thus have stopped reducing).

As such, we first need to define when we consider \lambdaAEff-computations 
to be in result form. 
It is important to note that for \lambdaAEff, 
the result forms have to also incorporate the \emph{temporary blocking} while computations await some promise (variable) $p$ to be fulfilled. 
Therefore, as a first step, we characterise such computations using the judgement $\awaiting p M$, given by the following three rules:
\begin{mathpar}
  \coopinfer{}{
  }{
    \awaiting p {\tmawait p x M}
  }

  \coopinfer{}{
    \awaiting p M
  }{
    \awaiting p {\tmlet x M N}
  }
  
  \coopinfer{}{
    \awaiting p M
  }{
    \awaiting p {\tmopin{op}{V}{M}}
  }
\end{mathpar}

Next, we characterise \lambdaAEff's \emph{result forms} using the judgements $\CompResult {\Psi} {M}$ and
$\RunResult {\Psi} {M}$:
\begin{mathpar}
  \coopinfer{}{
    \CompResult {\Psi} {M}
  }{
    \CompResult {\Psi} {\tmopout {op} V M}
  }
  \quad
  \coopinfer{}{
    \RunResult {\Psi} {M}
  }{
    \CompResult {\Psi} {M}
  }
  \vspace{-1ex}
  \\
  \coopinfer{}{
  }{
    \RunResult {\Psi} {\tmreturn V}
  }
  \quad
  \coopinfer{}{
    \RunResult {\Psi \cup \{p\}} {N}
  }{
    \RunResult {\Psi} {\tmwith {op} x M p N}
  }
  \quad
  \coopinfer{}{
    p \in \Psi \\
    \awaiting p M
  }{
    \RunResult {\Psi} {M}
  }
\end{mathpar}
In these judgements, $\Psi$ is a set  
of (promise) variables that have been bound by interrupt handlers enveloping the computation.
These judgements express that a computation $M$ is in a (top-level) 
result form $\CompResult {\Psi} {M}$ when, considered as a tree, it has a shape in which \emph{all}  
signals are towards the root, interrupt handlers are in the intermediate nodes, and 
the leaves contain return values and computations that are temporarily blocked
while awaiting one of the promise variables in $\Psi$ to be fulfilled. 
The slightly mysterious name of the intermediate judgement $\RunResult {\Psi} {M}$ will become clear
in \autoref{sec:basic-calculus:type-safety:processes}.
The finality of these result forms is captured by the next lemma.

\begin{lemma}
\label{lemma:results-are-final}
Given $\Psi$ and $M$ such that $\CompResult {\Psi} {M}$, then there exists no $N$ with $M \reduces N$.
\end{lemma}

We are now ready to state and prove the \emph{progress theorem} for the sequential part of \lambdaAEff.

\begin{theorem}
\label{theorem:progress}
Given a well-typed computation $\Gamma \types M : \tycomp{Y}{(\o,\i)}$, where $\Gamma = x_1 \of \typromise {X_1}, \ldots, x_n \of \typromise {X_n}$, 
then either (i) there exists an $N$ such that $M \reduces N$, or
(ii) we have $\CompResult {\{x_1, \ldots, x_n\}} {M}$.
\end{theorem}

\begin{proof}
The proof is standard and proceeds by induction 
on the derivation of  $\Gamma \types M : \tycomp{Y}{(\o,\i)}$. For instance, 
if the derivation ends with a typing rule for function application or pattern-matching, 
we use an auxiliary canonical forms lemma to show that the value involved 
is either a function abstraction or in constructor form---thus $M$ can $\beta$-reduce and we prove (i).
Here we crucially rely on the context $\Gamma$ having the specific 
assumed form $x_1 \of \typromise {X_1}, \ldots, x_n \of \typromise {X_n}$.
If the derivation ends with \textsc{TyComp-Await}, then we 
use a canonical forms lemma to show that the promise value is either a variable in
$\Gamma$, in which case we prove (ii), or in constructor form, in which case we prove (i).
If the derivation however ends with a typing rule for any of the terms figuring in the 
evaluation contexts $\E$, then we proceed based on using the induction hypothesis on the 
corresponding continuation.
\end{proof}

\begin{corollary}
\label{corollary:progress}
Given a well-typed closed computation $\types M : \tycomp{X}{(\o,\i)}$,  
then either (i) there exists a computation $N$ such that $M \reduces N$, or
(ii) $M$ is already in result form, i.e., we have $\CompResult {\emptyset} {M}$.
\end{corollary}

\subsubsection{Type Preservation}
\label{sec:basic-calculus:type-safety:preservation}

The type preservation result says that reduction preserves well-typedness.

The results that we present in this section use standard \emph{substitution  
lemmas}. For instance, given $\Gamma, x \of X , \Gamma' \types M : \tycomp{Y}{(\o,\i)}$
and $\Gamma \types V : X$, then we can show that $\Gamma, \Gamma' \types M[V/x] : \tycomp{Y}{(\o,\i)}$.
In the following we also use standard \emph{typing inversion lemmas}. For example, given 
$\Gamma \types \tmopin{op}{V}{M} : \tycomp{X}{(\o,\i)}$, then we can show that 
$\Gamma \types V : A_\op$ and $\Gamma \types M : \tycomp{X}{\opincomp {op} (\o',\i')}$, 
such that $\opincomp {op} (\o',\i') \order {O \times I} (\o,\i)$.

As the proof of type preservation proceeds by induction on reduction steps, 
we find it useful to define an auxiliary \emph{typing judgement for evaluation contexts}, 
written $\Gamma \types\!\![\, \Gamma' \,\vert\, \tycomp{X}{(\o,\i)} \,]~ \E : \tycomp{Y}{(\o',\i')}$, 
which we then use to prove the evaluation context rule case of the proof.
Here, $\Gamma'$ is the context of variables bound by the interrupt handlers in $\E$, and 
$\tycomp{X}{(\o,\i)}$ is the type of the hole $[~]$. This judgement is given using rules 
similar to those for computations, including subtyping, e.g., we have
\begin{mathpar}
  \coopinfer{}{
    \i'\, (\op) = (\o'',\i'') \\
    \Gamma, x \of A_\op \types M : \tycomp{\typromise Y}{(\o'',\i'')} \\
    \Gamma, p \of \langle Y \rangle \types\!\![\, \Gamma' \,\vert\, \tycomp{X}{(\o,\i)} \,]~ \E : \tycomp{Z}{(\o',\i')}
  }{
    \Gamma \types\!\![\, p \of \langle Y \rangle, \Gamma' \,\vert\, \tycomp{X}{(\o,\i)} \,]~ \tmwith{op}{x}{M}{p}{\E} : \tycomp{Z}{(\o',\i')}
  }
\end{mathpar}
It is thus straightforward to relate this typing of evaluation contexts with that of computations.

\begin{lemma}
\label{lemma:eval-ctx-typing}
\mbox{}

$\Gamma \types \E[M] : \tycomp{Y}{(\o',\i')} 
\Leftrightarrow 
\exists\, \Gamma', X, \o, \i .~ 
\Gamma \types\!\![\, \Gamma' \,\vert\, \tycomp{X}{(\o,\i)} \,]~ \E : \tycomp{Y}{(\o',\i')}
~\wedge~
\Gamma,\Gamma' \types M : \tycomp{X}{(\o,\i)}
$.
\end{lemma}

We are now ready to state and prove the \emph{type preservation theorem} for the sequential part of \lambdaAEff.

\begin{theorem}
\label{theorem:preservation}
Given $\Gamma \types M : \tycomp{X}{(\o,\i)}$ and 
$M \reduces N$, then we have $\Gamma \types N : \tycomp{X}{(\o,\i)}$.
\end{theorem}

\begin{proof}
The proof is standard and proceeds by induction on the derivation of  
$M \reduces N$, using typing inversion lemmas depending on 
the structure forced upon $M$ by the last rule used in $M \reduces N$.

There are four cases of interest in this proof. The first two concern
the interaction of incoming interrupts and interrupt handlers.
On the one hand, if the given derivation of $\reduces$ ends with 
\[
\tmopin{op}{V}{\tmwith{op}{x}{M}{p}{N}} \reduces \tmlet{p}{M[V/x]}{\tmopin{op}{V}{N}}
\]
then in order to type the right-hand side of this rule, we are led to use subtyping with  
\srefcase{Lemma}{lemma:action}{2}, so as to show that $M$'s effect information is 
included in $\opincomp {op} {(\o , \i)}$. On the other hand, given
\[
\tmopin{op'}{V}{\tmwith{op}{x}{M}{p}{N}} \reduces \tmwith{op}{x}{M}{p}{\tmopin{op'}{V}{N}}
\quad {\color{rulenameColor}(\op \neq \op')}
\]
then in order to type the right-hand side of this rule, we are led to use subtyping with  
\srefcase{Lemma}{lemma:action}{3}, so as to show that 
after acting on $(\o,\i)$ with $\op'$, $\op$ remains mapped to $M$'s effect information.

The third case of interest concerns the commutativity of signals with interrupt handlers:
\[
\tmwith{op}{x}{M}{p}{\tmopout{op'}{V}{N}} \reduces \tmopout{op'}{V}{\tmwith{op}{x}{M}{p}{N}}
\]
where in order to type the signal's payload $V$ in the right-hand side, 
it is crucial that the promise-typed variable $p$ cannot appear in $V$---this is ensured by 
our type system that restricts the signatures $\op : A_\op$ to ground types. As a result, 
we can strengthen the typing context of $V$ by removing $p$.

Finally, in the evaluation context rule case, we use the induction hypothesis with \sref{Lemma}{lemma:eval-ctx-typing}.
\end{proof}

Interestingly, the proof of \autoref{theorem:preservation} tells us that if 
one were to consider a variant of \lambdaAEff~in which the 
\textsc{TyComp-Subsume} rule appeared as an explicit coercion term $\tmkw{coerce}_{(\o,\i) \order {O \times I} (\o',\i')}\, M$, then 
the right-hand sides of the two interrupt propagation rules highlighted in the above proof 
would also need to involve such coercions, corresponding to the uses of \sref{Lemma}{lemma:action}. 
This however means that other computations involved in these reduction rules would also need to be type-annotated.

% !TEX root = paper.tex

\section{A Calculus for Asynchronous Effects: Parallel Processes}
\label{sec:basic-calculus:processes}

Next, we describe the parallel part of \lambdaAEff. Similarly to the sequential part, we 
again present the corresponding syntax, a small-step semantics, 
a type-and-effect system, and type safety results.

\subsection{Parallel Processes}

To keep the presentation focussed on the asynchronous use of algebraic effects, we 
consider a very simple model of parallelism: a process is either an \emph{individual computation} 
or the \emph{parallel composition} of two processes. To facilitate interactions between processes, they also  
contain outward propagating \emph{signals} and inward propagating \emph{interrupts}. Formally, the 
syntax of \emph{parallel processes} is
\[
  P, Q
  \bnfis \tmrun M
  \,\bnfor\! \tmpar P Q
  \,\bnfor\! \tmopout{op}{V}{P}
  \,\bnfor\! \tmopin{op}{V}{P}
\]
Note that processes do not include interrupt handlers---these are local to individual computations.

We leave first-class processes and their dynamic creation for future work, as discussed in \autoref{sec:conclusion}.

\subsection{Small-Step Operational Semantics}

We equip the parallel part of \lambdaAEff~with a small-step semantics that  
naturally extends that of \lambdaAEff's sequential part.
The semantics is defined using a reduction relation $P \reduces Q$, as given in \autoref{fig:processes}.

\begin{figure}[tp]
  \parbox{\textwidth}{
  \centering
  \small
  \begin{minipage}[t]{0.4\textwidth}
  \centering
  \begin{align*}
  \intertext{\textbf{Individual computations}}
    \coopinfer{}{
      M \reduces N
    }{
      \tmrun M \reduces \tmrun N
    }
  \end{align*}
  \begin{align*}
    \intertext{\textbf{Signal hoisting}}
    \tmrun {(\tmopout{op}{V}{M})}  &\reduces \tmopout{op}{V}{\tmrun M}
    \\[1ex]
    \intertext{\textbf{Broadcasting}}
    \tmpar{\tmopout{op}{V}{P}}{Q} &\reduces \tmopout{op}{V}{\tmpar{P}{\tmopin{op}{V}{Q}}}
    \\
    \tmpar{P}{\tmopout{op}{V}{Q}} &\reduces \tmopout{op}{V}{\tmpar{\tmopin{op}{V}{P}}{Q}}
  \end{align*}
  \vspace{-1ex}
  \end{minipage}
  \qquad
  \begin{minipage}[t]{0.4\textwidth}
  \centering
  \begin{align*}
    \intertext{\textbf{Interrupt propagation}}
    \tmopin{op}{V}{\tmrun M} &\reduces \tmrun {(\tmopin{op}{V}{M})}
    \\
    \tmopin{op}{V}{\tmpar P Q} &\reduces \tmpar {\tmopin{op}{V}{P}} {\tmopin{op}{V}{Q}}
    \\
    \tmopin{op}{V}{\tmopout{op'}{W}{P}} &\reduces \tmopout{op'}{W}{\tmopin{op}{V}{P}}
  \end{align*}
  \begin{align*}
    \intertext{\quad\textbf{Evaluation context rule}}
    \quad
    \coopinfer{}{
      P \reduces Q
    }{
      \F[P] \reduces \F[Q]
    }
  \end{align*}
  \end{minipage}
  \begin{align*}
  \intertext{\textbf{where}\vspace{1ex}}
  \text{$\F$}
  \bnfis& [~]
  \bnfor \tmpar \F Q \bnfor\! \tmpar P \F
  \bnfor \tmopout{op}{V}{\F}
  \bnfor \tmopin{op}{V}{\F}
  \end{align*}
  } 
  \caption{Small-step operational semantics of parallel processes.}
  \label{fig:processes}
\end{figure}

\paragraph{Individual computations}
This reduction rule states that, as processes, individual computations evolve according to the small-step
operational semantics $M \reduces N$ we defined for them in \autoref{sec:basic-calculus:semantics:computations}.

\paragraph{Signal hoisting}
This rule propagates signals out of individual computations.
It is important to note that we only hoist those signals that have propagated to the outer boundary
of a computation.

\paragraph{Broadcasting}
The broadcast rules turn outward moving signals in one process into inward moving interrupts 
for the process parallel to it, while continuing to propagate the signals outwards to any 
further parallel processes. The latter ensures that the semantics is compositional.

\paragraph{Interrupt propagation}
These three rules simply propagate interrupts inwards into individual computations, 
into all branches of parallel compositions, and past any outward moving signals.

\paragraph{Evaluation contexts}
Analogously to the semantics of computations, the semantics of processes also includes a context rule, which allows reductions under \emph{evaluation contexts} 
$\F$. Observe that compared to the evaluation contexts for computations, those for processes
do not bind variables. 

\subsection{Type-and-Effect System}

Analogously to its sequential part, we also equip \lambdaAEff's parallel part with a type-and-effect system.

\paragraph{Types} The \emph{types of processes} are designed to match their parallel structure---they are given by
\[
  \text{$\tyC$, $\tyD$}
  \bnfis \tyrun X \o \i
  \,\bnfor\! \typar \tyC \tyD
\]
Intuitively, $\tyrun X \o \i$ is a process type of an individual computation of type $\tycomp{X}{(\o,\i)}$, and $\typar \tyC \tyD$
is the type of the parallel composition of two processes that respectively have types $\tyC$ and $\tyD$.

\paragraph{Typing judgements}
\emph{Well-typed processes} are characterised using the judgement
$\Gamma \vdash P : \tyC$. We present the typing rules in \autoref{fig:process-typing-rules}.
While our processes are not currently higher-order, we allow 
non-empty contexts $\Gamma$ to model the possibility of using libraries and top-level function definitions.

\begin{figure}[tp]
  \centering
  \small
  \begin{mathpar}
  \coopinfer{TyProc-Run}{
    \Gamma \types M : \tycomp{X}{(\o,\i)}
  }{
    \Gamma \types \tmrun{M} : \tyrun{X}{\o}{\i}
  }
  \quad
  \coopinfer{TyProc-Par}{
    \Gamma \types P : \tyC \\
    \Gamma \types Q : \tyD
  }{
    \Gamma \types \tmpar{P}{Q} : \typar{\tyC}{\tyD}
  }
  \quad
  \coopinfer{TyProc-Signal}{
    \op \in \mathsf{signals\text{-}of}{(\tyC)} \\\\
    \Gamma \types V : A_\op \\
    \Gamma \types P : \tyC 
  }{
    \Gamma \types \tmopout{op}{V}{P} : \tyC
  }
  \quad
  \coopinfer{TyProc-Interrupt}{
    \Gamma \types V : A_\op \\
    \Gamma \types P : \tyC 
  }{
    \Gamma \types \tmopin{op}{V}{P} : \opincomp{op}{\tyC}
  }  
  \end{mathpar}
  \caption{Process typing rules.}
  \label{fig:process-typing-rules}
\end{figure}

The rules \textsc{TyProc-Run} and \textsc{TyProc-Par} capture the earlier 
intuition about the types of processes matching their parallel structure. The rules 
\textsc{TyProc-Signal} and \textsc{TyProc-Interrupt} are similar to the corresponding rules 
from \autoref{fig:computation-typing-rules}.
The \emph{signal annotations} of a process type are calculated as
\[
\mathsf{signals\text{-}of}(\tyrun{X}{\o}{\i}) ~\defeq~ \o
\qquad\qquad
\mathsf{signals\text{-}of}(\typar{\tyC}{\tyD}) ~\defeq~ \mathsf{signals\text{-}of}(\tyC) \cup \mathsf{signals\text{-}of}(\tyD)
\]
and the \emph{action of interrupts} on process types $\opincomp{op}{\tyC}$ extends the action on effect annotations as
\[
\opincomp{op}(\tyrun{X}{\o}{\i}) 
~\defeq~
X \att (\opincomp {op} {(\o , \i)})
\qquad\qquad
\opincomp{op}(\typar{\tyC}{\tyD}) 
~\defeq~
\typar{(\opincomp{op}{\tyC})}{(\opincomp{op}{\tyD})}
\]
by propagating the interrupt towards the types of individual computations. 
We then have:

\begin{lemma}
\label{lemma:signals-of-interrupt-action}
For any process type $\tyC$ and interrupt $\op$, we have that $\mathsf{signals\text{-}of}(\tyC) \order O \pi_1\, (\opincomp{op}{\tyC})$.
\end{lemma}

It is worth noting that \autoref{fig:process-typing-rules} does not include an analogue  
of \textsc{TyComp-Subsume}. This is 
deliberate because as we shall see below, \emph{process types reduce}
in conjunction with the processes they are assigned to, and the outcome   
is generally neither a sub- nor supertype of the original type.

\subsection{Type Safety}
\label{sec:basic-calculus:type-safety:processes}

We conclude the meta-theory of \lambdaAEff~by proving type safety 
for its parallel part. Analogously to \autoref{sec:basic-calculus:type-safety}, 
we once again split type safety into separate proofs of \emph{progress} 
and \emph{preservation}.

\subsubsection{Progress}

We characterise the \emph{result forms} of parallel processes 
by defining two judgements, $\ProcResult P$ and $\ParResult P$, 
and by using the judgement $\RunResult {\Psi} {M}$ from 
\autoref{sec:basic-calculus:type-safety}, as follows:
\begin{mathpar}
  \coopinfer{}{
    \ProcResult {P}
  }{
    \ProcResult {\tmopout {op} V P}
  }
  \qquad
  \coopinfer{}{
    \ParResult {P}
  }{
    \ProcResult {P}
  }
  \qquad
  \coopinfer{}{
    \RunResult {\emptyset} {M}
  }{
    \ParResult {\tmrun M}
  }
  \qquad
  \coopinfer{}{
    \ParResult P \\
    \ParResult Q
  }{
    \ParResult {\tmpar P Q}
  }
\end{mathpar}
These judgements express that a process $P$ is in a (top-level) 
result form $\ProcResult {P}$ when, considered as a tree, it has a shape in which 
\emph{all} signals are towards the root, parallel compositions are in 
the intermediate nodes, and individual computation results are at the leaves. 
Importantly, the computation results $\RunResult {\emptyset} {M}$ we use here are those from 
which signals have been propagated out of 
(see \autoref{sec:basic-calculus:type-safety:progress}). 
The finality of these results forms is then captured by the next lemma.

\begin{lemma}
\label{lemma:results-are-final:processes}
Given a process $P$ such that $\ProcResult {P}$, then there exists no $Q$ such that $P \reduces Q$.
\end{lemma}

We are now ready to state and prove the \emph{progress theorem} for the parallel part of \lambdaAEff.

\begin{theorem}
Given a well-typed closed process $\types P : \tyC$,  
then either (i) there exists a process $Q$ such that $P \reduces Q$, or
(ii) the process $P$ is already in a (top-level) result form, i.e., we have $\ProcResult {P}$.
\end{theorem}

\begin{proof}
The proof is standard and proceeds by induction on the derivation of $\types P : \tyC$. 
In the base case, when the derivation ends with the \textsc{TyProc-Run} rule, 
and $P \hspace{-0.05cm}=\hspace{-0.05cm} \tmrun {\hspace{-0.05cm}M}$, we use  
\sref{Corollary}{corollary:progress}.
\end{proof}

\subsubsection{Type Preservation}
First, we note that the broadcast rules in \autoref{fig:processes} introduce new 
inward propagating interrupts in their right-hand sides that originally do not exist in their left-hand sides. As a result, 
compared to the types one assigns to the left-hand sides of these reduction rules, the types assigned to 
their right-hand sides will need to feature corresponding type-level actions of these interrupts.
We formalise this idea using a \emph{process type reduction} relation $\tyC \tyreduces \tyD$, given by
\[
  \coopinfer{}{
  }{
    \tyrun{X}{\o}{\i} \tyreduces \tyrun{X}{\o}{\i} 
  }
  \quad
  \coopinfer{}{
  }{
    X \att \opincompp {ops} {(\o , \i)} \tyreduces X \att \opincompp {ops} {(\opincomp {op} {(\o , \i)})}
  }
  \quad
  \coopinfer{}{
    \tyC \tyreduces \tyC' \\
    \tyD \tyreduces \tyD'
  }{
    \typar{\tyC}{\tyD} \tyreduces \typar{\tyC'}{\tyD'}
  }
\]
where we write $\opincompp {ops} {(\o , \i)}$ for a recursively defined \emph{action of a list of interrupts} on $(\o , \i)$, 
given by
\[
\opincompp {[]} {(\o , \i)} ~\defeq~ (\o , \i)
\qquad
\opincompp {(\op :: \opsym{ops})} {(\o , \i)} ~\defeq~ \opincomp {op} {(\opincompp {ops} (\o , \i))}
\]
Intuitively, $\tyC \tyreduces \tyD$ describes how process types reduce by being acted upon by 
freshly arriving interrupts. While we define the action behaviour only at the leaves of process types (under some 
enveloping sequence of actions), we can prove expected properties for arbitrary process types:

\begin{lemma}
\label{lemma:type-reduction} \mbox{}
\begin{enumerate}
\item Process types can remain unreduced, i.e., $\tyC \tyreduces \tyC$ for any process type $\tyC$.
\item Process types reduce by being acted upon, i.e., $\tyC \tyreduces \opincomp {op} \tyC$ for any type $\tyC$ and interrupt $\op$.
\item Process types can reduce under enveloping actions, i.e., $\opincomp {op} \tyC \tyreduces \opincomp {op} \tyD$ when $\tyC \tyreduces \tyD$.
\item Process type reduction can introduce signals, i.e., $\mathsf{signals\text{-}of} (\tyC) \order O \mathsf{signals\text{-}of} (\tyD)$
when $\tyC \tyreduces \tyD$.
\end{enumerate}
\end{lemma}

For the proof of \srefcase{Lemma}{lemma:type-reduction}{3}, it is important that we 
introduce interrupts under an arbitrary enveloping sequence of interrupt actions, 
and not simply as 
$X \att {(\o , \i)} \tyreduces X \att (\opincomp {op} {(\o , \i)})$.
Further, the proof of \srefcase{Lemma}{lemma:type-reduction}{4} requires us to generalise \srefcase{Lemma}{lemma:action}{1} to lists of enveloping actions:

\begin{lemma}
\label{lemma:signal-inclusion-lists-of-interrupts}
$\pi_1\, (\opincompp {ops} {(\o,\i)}) \order O \pi_1\, (\opincompp {ops} {(\opincomp {op} {(\o,\i)})})$
\end{lemma}

As in \autoref{sec:basic-calculus:type-safety:preservation}, we again find it useful 
to define a separate \emph{typing judgement for evaluation contexts}, this time written 
$\Gamma \types\!\![\, \tyC \,]~ \F : \tyD$, together with an 
analogue of \sref{Lemma}{lemma:eval-ctx-typing}, which we omit here. Instead, we 
observe that this typing judgement is subject to process type reduction:

\begin{lemma}
\label{lemma:hoisting-and-evaluation-context-types}
Given $\Gamma \types\!\![\, \tyC \,]~ \F \hspace{-0.05cm}:\hspace{-0.05cm} \tyD$ and $\tyC \hspace{-0.05cm}\tyreduces\hspace{-0.05cm} \tyC'$, then there exists $\tyD'$ with 
$\tyD \hspace{-0.05cm}\tyreduces\hspace{-0.05cm} \tyD'$ and $\Gamma \types\!\![\, \tyC' \,]~ \F \hspace{-0.05cm}:\hspace{-0.05cm} \tyD'$.
\end{lemma}

We are now ready to state and prove the \emph{type preservation theorem} for the parallel part of \lambdaAEff.

\begin{theorem}
\label{theorem:preservation:processes}
Given a well-typed process $\Gamma \types P : \tyC$, such that $P$ can reduce as 
$P \reduces Q$, then there exists a process type $\tyD$, such 
that the process type $\tyC$ can reduce as $\tyC \tyreduces \tyD$, and we have $\Gamma \types Q : \tyD$.
\end{theorem}

\begin{proof}
The proof proceeds by induction on the derivation of  
$P \reduces Q$, using auxiliary typing inversion lemmas depending on 
the structure forced upon $P$ by the last rule used in $P \reduces Q$.
For all but the broadcast and evaluation context rules, we can pick $\tyD$ to be $\tyC$ and use 
\srefcase{Lemma}{lemma:type-reduction}{1}.
For the broadcast rules, we define $\tyD$ by introducing the corresponding 
interrupt, and build $\tyC \tyreduces \tyD$ using the parallel composition 
rule together with \srefcase{Lemma}{lemma:type-reduction}{2}.
For the evaluation context rule, we use \sref{Lemma}{lemma:hoisting-and-evaluation-context-types}
in combination with the induction hypothesis. Finally, in order to discharge effects-related side-conditions 
when commuting interrupts with signals, 
we use \sref{Lemma}{lemma:signals-of-interrupt-action}.
\end{proof}

% !TEX root = paper.tex

\section{Asynchronous Effects in Action}
\label{sec:applications}

We now show some examples of the kinds of programs one can write in \lambdaAEff.
Similarly to \autoref{sec:overview:runningexample},   
we again allow ourselves access to mutable references, and use 
generic versions $\tmopoutgen {op} V$ of signals.

\subsection{Guarded Interrupt Handlers}
\label{sec:applications:guarder-handlers}

Before diving into the examples, we note that we often want the 
triggering of interrupt handlers to be 
based on not only the names of interrupts, but also the payloads that they carry.
In order to express such more fine-grained triggering behaviour, we shall use a
\emph{guarded interrupt handler}: 
\begin{lstlisting}
promise (op x when guard |-> comp) as p in cont
\end{lstlisting}
which is simply a syntactic sugar for the following interrupt handler that recursively reinstalls 
itself until the boolean \ls$guard$ becomes true, in which case it executes the handler code \ls$comp$:
\begin{lstlisting}
let rec waitForGuard () =
    promise (op x |-> if guard then comp else waitForGuard ()) as p' in return p'
in
let p = waitForGuard () in cont
\end{lstlisting}
Here, \ls$x$ is bound both in \ls$guard$ and \ls$comp$. Further, if \ls$comp$ has type $\tycomp{\typromise X}{(\o',\i')}$
and \ls$cont$ has type $\tycomp{Y}{(\o,\i)}$, such that $\i\, (\op) = (\o',\i')$, then we can assign the entire computation
the type $\tycomp{Y}{(\o,\i \cup \i_h)}$, where the effect annotation $\i_h$ is the least fixed point of the map
$\i'' \mapsto \{ \op \mapsto (\o',\i' \cup \i'') \} : I \to I$. Observe that some of the recursive encoding leaks into the 
type of the entire computation via $\i_h$.

Note that regardless whether \ls$guard$ is true, every interrupt is propagated into \ls$cont$.
To typecheck their definition, and to ensure that guarded interrupt handlers
are non-blocking, it is crucial that the handler code of 
ordinary interrupt handlers returns promise-typed values, as noted in \autoref{sect:typing-rules}.

\subsection{Preemptive Multi-Threading}
\label{sec:applications:multithreading}

Multi-threading remains one of the most exciting applications of algebraic effects, with the possibility of expressing 
many evaluation strategies being the main reason for the extension of \pl{Multicore OCaml} with effect handlers~\cite{Dolan:MulticoreOCaml}.
These evaluation strategies are however \emph{cooperative} in nature, where each thread needs to explicitly yield back 
control, stalling other threads until then. 

While it is possible to also simulate \emph{preemptive multi-threading} within the conventional treatment of algebraic effects, 
it requires a low-level access to the specific runtime environment, so as to inject 
yields into the currently running computation~\cite{Dolan:MulticoreOCaml}.
In contrast, implementing preemptive multi-threading in \lambdaAEff~is quite straightforward, and importantly, 
possible within the language itself---the injections into the running computation 
take the form of incoming interrupts.

For this, let us consider two interrupts, $\opsym{stop} : \tyunit$ and $\opsym{go} : \tyunit$, that communicate to a thread whether to \emph{pause} or 
\emph{resume} execution. These interrupts can originate from a timer process we run in parallel.

At the core of our implementation of preemptive multi-threading is
the recursive function
\begin{lstlisting}
let rec waitForStop () =
    promise (stop _ |->
        promise (go _ |-> return <<()>>) as p in (await p until <<_>> in waitForStop ())
    ) as p' in return p'
\end{lstlisting}
which first installs an interrupt handler for $\opsym{stop}$, letting subsequent computations run their course. Once 
the $\opsym{stop}$ interrupt arrives, the interrupt handler for it is triggered and the next one for $\opsym{go}$ gets 
installed. In contrast to the interrupt handler for $\opsym{stop}$, the one for $\opsym{go}$ starts awaiting 
the (unit) promise \ls$p$. This means that any subsequent computations are blocked until a $\opsym{go}$ interrupt 
is received, after which we recursively reinstall the interrupt handler for $\opsym{stop}$ and repeat the cycle.

To \emph{initiate the preemptive behaviour} for some computation \ls{comp}, we simply run the program
\begin{lstlisting}
waitForStop (); comp
\end{lstlisting}
The algebraicity reduction rules for interrupt handlers ensure that they propagate out of \ls{waitForStop} 
and encompass the entire computation, including \ls{comp}.
Observe that in contrast to usual effect handler based encodings of multi-threading, \ls$waitForStop$ does 
not need any access to a thunk \linebreak \lstinline{fun () |-> comp} representing the threaded computation. 
In particular, the given computation \ls$comp$ can be completely unaware of the multi-threaded behaviour, 
both in its definition and its type.

This approach can be easily extended to multiple threads, by using interrupts' payloads to communicate thread IDs. To this end, 
we can consider interrupts $\opsym{stop} : \tyint$ and $\opsym{go} : \tyint$, and define 
\begin{lstlisting}
let rec waitForStop threadID =
    promise (stop threadID' when threadID = threadID' |->
        promise (go threadID' when threadID = threadID' |-> return <<()>>) as p in
        await p until <<_>> in waitForStop threadID
    ) as p' in return p'
\end{lstlisting}
using guarded interrupt handlers, and 
conditioning their triggering based on the received IDs.

\subsection{Remote Function Calls}
\label{sec:applications:remotecall}

One of the main uses of asynchronous computation is to offload the execution of 
\emph{long-running functions} $f \!: \tyfun{A}{\tycomp{B}{(\o,\i)}}$ to remote processes. 
Below we show how to implement this in \lambdaAEff.

One invokes a remote function by issuing a signal named $\opsym{call}$ with the \emph{function's 
argument}, and then awaits an interrupt named $\opsym{result}$ with the \emph{function's result},  
with all effects specified by $(\o,\i)$ happening at the callee site. The caller then calls 
such a remote function through a wrapper \ls$callWith$, 
which issues the $\opsym{call}$ signal, installs a handler for the $\opsym{result}$ 
interrupt, and returns a thunk that awaits the function's result.
For instance, one may then use remote functions in their code as
\begin{lstlisting}
let subtally = callWith "SELECT count(col) FROM table WHERE cond" in
let tally = callWith "SELECT count(col) FROM table" in
printf "Percentage: %d" (100 * subtally () / tally ())
\end{lstlisting}

To avoid the results of earlier remote function calls from fulfilling the promises of 
later ones, we assign to each call a unique identifier, and communicate 
those in payloads. We implement these unique identifiers using a counter.
For a remote function $f : \tyfun{A}{\tycomp{B}{(\o,\i)}}$, we 
type the signals and interrupts as $\opsym{call} : \typrod{A}{\tysym{int}}$ and $\opsym{result} : \typrod{B}{\tysym{int}}$.
The \emph{caller site} function \ls$callWith$ is then defined as
\begin{lstlisting}
let callWith x =
    let callNo = !callCounter in callCounter := !callCounter + 1;
    send call (x, callNo);
    promise (result (y, callNo') when callNo = callNo' |-> return <<y>>) as resultPromise in
    return (fun () -> await resultPromise until <<resultValue>> in return resultValue)
\end{lstlisting}
After issuing the $\opsym{call}$ signal, \ls$callWith$ installs a guarded interrupt handler 
for the corresponding $\opsym{result}$ interrupt, and then returns a function that, 
when called, awaits the result of the remote call.

At the \emph{callee site}, we simply install an interrupt handler that executes the function in question, 
issues an outgoing signal with the function's result, and then recursively reinstalls itself, as follows:
\begin{lstlisting}
let remote f =
    let rec loop () =
        promise (call (x, callNo) |-> let y = f x in send result (y, callNo); loop ()) as p in return p
    in loop ()
\end{lstlisting}

Unlike effect handlers, our interrupt handlers have very limited control over the execution of 
their continuation. However, we can still simulate \emph{cancellations of asynchronous 
computations} by awaiting a promise that will never be fulfilled. We achieve this with the help of the function
\begin{lstlisting}
let awaitCancel callNo runBeforeStall =
    promise (cancel callNo' when callNo = callNo' |->
        promise (dummy () |-> return <<()>>) as dummyPromise in
        runBeforeStall ();
        await dummyPromise until <<_>> return <<()>>
    ) as _ in return ()
\end{lstlisting}
which takes the identifier of the remote function call that we want to make cancellable, and a thunked 
computation to run before the continuation is stalled. We can then extend 
the callee site with cancellable function calls by invoking \lstinline{awaitCancel} before 
we start executing the long-running computation \lstinline{f x}. In particular, 
we change the interrupt handler code in \ls$remote f$ to read as follows: 
\begin{lstlisting}
  call (x, callNo) |-> awaitCancel callNo loop; let y = f x in send result (y, callNo); loop ()
\end{lstlisting}

However, if left as is, cancelling one call would cancel all
unfinished remote function calls because they would be part of the stalled continuation. 
To overcome this, we run the callee site in parallel with an auxiliary process (which we omit here) that
reacts to a $\opsym{cancel}$ interrupt by \emph{reinvoking these unfinished calls} (minus the cancelled one) 
by reissuing the corresponding $\opsym{call}$ signals, which then get propagated to the callee site, and to  
the \lstinline{loop ()} we run in \lstinline{awaitCancel callNo loop} before stalling.

We note that the cancelled computation is only \emph{perpetually stalled}, but not discarded completely, leading 
to a memory leak. We conjecture that extending \lambdaAEff~with effect handlers that  
have greater control over the continuation could lead to a more efficient code for the callee site.
We also conjecture that a future extension of \lambdaAEff~with dynamic process creation would 
eliminate the need for the auxiliary reinvoker process, because then the callee site could create a 
new process for every remote function call it receives, and each $\opsym{cancel}$ interrupt would 
stall only one of such (sub-)processes. 

\subsection{Runners of Algebraic Effects}
\label{sec:applications:runners}

Next, we use \lambdaAEff~to implement a parallel variant 
of \emph{runners of algebraic effects} \cite{Ahman:Runners}. These are a 
natural mathematical model and programming abstraction for resource management based on 
algebraic effects, and correspond to effect handlers that apply continuations (at most) 
once in a tail call position.
In a nutshell, for a signature of operation 
symbols $\op : A_\op \to B_\op$, a \emph{runner} $\mathcal{R}$ comprises a family of stateful functions 
$\overline{\op}_{\mathcal{R}} : A_\op \to R \Rightarrow B_\op \times R$, 
called \emph{co-operations}, where $R$ is the type of \emph{resources} that the runner manipulates.
In the more general setting of \citet{Ahman:Runners}, the co-operations also model other, external 
effects, such as native calls to the operating system, and can furthermore raise  
exceptions---all of which we shall gloss over here.

Given a runner $\mathcal{R}$, \citet{Ahman:Runners} provide the programmer with a construct
\vspace{-0.05em}
\[
\tmkw{using}~\mathcal{R}~\tmkw{@}~V_{\text{init}}~\tmkw{run}~M~\tmkw{finally}~\{ \tmreturn x ~\tmkw{@}~ r_{\text{fin}} \mapsto N  \}
\vspace{-0.05em}
\]
which runs $M$ using $\mathcal{R}$, with resources initially set to $V_{\text{init}}$; and 
finalises the return value and final resources using $N$, e.g., ensuring that all file handles get closed.
This is a form of effect handling: it executes $M$ by invoking 
co-operations in place of operation calls, while doing resource-passing under the hood. 
Below we show by means of examples how one can use \lambdaAEff~to naturally separate $\mathcal{R}$ and $M$ 
into different processes.
For simplicity, we omit the initialisation and finalisation phases.

For our first example, let us consider a runner that implements a \emph{pseudo-random number generator} by providing a co-operation for ${\opsym{random} : \tyunit \to \tyint}$, which we can for example implement as 
\begin{lstlisting}
let linearCongruenceGeneratorRunner modulus a c initialSeed =
    let rec loop seed =
        promise (randomReq callNo |->
            let seed' = (a * seed + c) mod modulus in send randomRes (seed, callNo); loop seed'
        ) as p in return p
    in loop initialSeed
\end{lstlisting}
It is given by a recursive interrupt handler, which listens for $\opsym{randomReq} : \tysym{int}$ requests 
issued by clients, and itself issues $\opsym{randomRes} : \typrod{\tysym{int}}{\tysym{int}}$ responses. The resource this runner 
manages is the seed, which it passes between subsequent co-operation calls as an argument to the recursive \ls$loop$.

For the client code $M$, we implement operation calls \ls$random ()$ as discussed in \autoref{sec:overview:signals}, 
by decoupling them into signals and interrupts. We again use guarded interrupt handlers and call identifiers to 
avoid a response to one operation call fulfilling the promises of subsequent ones. 

\begin{lstlisting}
  let random () =
      let callNo = !callCounter in callCounter := callNo + 1;
      send randomReq callNo;
      promise (randomRes (n, callNo') when callNo = callNo' |-> return <<n mod 10>>) as p in
      await p until <<m>> in return m
\end{lstlisting}

As a second example, we show that this parallel approach to runners naturally 
extends to multiple co-operations. Specifically, we implement a \emph{runner for a heap}, 
by providing co-operations for 
\vspace{-0.05em}
\[
    \opsym{alloc} : \tysym{int} \to \tysym{loc} \qquad
    \opsym{lookup} : \tysym{loc} \to \tysym{int} \qquad
    \opsym{update} : \tysym{loc} \times \tysym{int} \to \tyunit
\vspace{-0.05em}
\]
We represent these co-operations using a signal/interrupt pair $(\opsym{opReq},\opsym{opRes})$ with payload types

\begin{lstlisting}
type payloadReq = | AllocReq of int | LookupReq of loc | UpdateReq of loc * int
type payloadRes = | AllocRes of loc | LookupRes of int | UpdateRes of unit
\end{lstlisting}

\noindent
The resulting runner is then implemented by pattern-matching on the payload value as follows:
\begin{lstlisting}
let rec heapRunner heap =
    promise (opReq (payloadReq, callNo) |->
        let heap', payloadRes =
            match payloadReq with
            | AllocReq v |-> let heap', l = allocHeap heap v in return (heap', AllocRes l)
            | LookupReq l |-> let v = lookupHeap heap l in return (heap, LookupRes v)
            | UpdateReq (l, v) |-> let heap' = updateHeap heap l v in return (heap', UpdateRes ())
        in
        send opRes (payloadRes, callNo); heapRunner heap'
    ) as p in return p
\end{lstlisting}
Note that by storing \ls$heap$ in memory, we could have also used three signal/interrupt pairs and 
split \ls$heapRunner$ into three distinct interrupt handlers, one for each of allocation, lookup, and update.

\subsection{Non-Blocking Post-Processing of Promised Values}
\label{sec:applications:chaining}

As discussed in \autoref{sect:overview:promising}, interrupt handlers differ
from ordinary operation calls by allowing user-side post-processing of received data. 
In this final example, we show that \lambdaAEff~is flexible enough to modularly perform \emph{further 
non-blocking post-processing} of this data anywhere in a program.

For instance, let us assume we are writing a program that contains an interrupt handler (for some $\op$)
that promises to return us a list of integers. Now, at some later point 
in the program, we decide that we want to further process this list 
if and when it becomes available, 
e.g., by using some of its elements to issue an outgoing signal. 
Of course, we could do this by going back and changing the original interrupt handler, 
but this would not be very modular; 
nor do we want to block the entire program's execution (using \ls$await$) until $\op$ 
arrives and the concrete list becomes available.

Instead, we can define a generic combinator for \emph{non-blocking post-processing} of  
promised values
\begin{lstlisting}
process$_{\op}$ p with (<<x>> |-> comp) as q in cont
\end{lstlisting}
that takes an earlier made promise \ls$p$ (which we assume originates 
from handling the specified interrupt $\op$), 
and makes a new promise to
execute the post-processing code \ls$comp[v/x]$ once \ls$p$ gets fulfilled with some value \ls$v$. The (non-blocking) continuation
\ls$cont$ can refer to \ls$comp$'s result using the new promise-typed variable \ls$q$ bound 
in it. Under the hood, \ls{process$_{\op}$} is a syntactic sugar for
\begin{lstlisting}
promise (op _ |-> await p until <<x>> in let y = comp in return <<y>>) as q in cont
\end{lstlisting}
While \ls{process$_{\op}$} does involve an \ls$await$, it gets 
exposed only after  
\ls$op$ is received, but by that time \ls$p$ will have been fulfilled with some  
\ls$v$ by an earlier interrupt handler, and thus the \ls$await$ can reduce.

Returning to post-processing a list of integers promised by an existing interrupt handler, 
below is an example showing the use of the \ls{process$_{\op}$} combinator and how to \emph{chain multiple 
post-processing computations together} (here, filtering, folding, and issuing an outgoing signal), 
in the same spirit as how one is taught to program compositionally with futures and promises~\cite{Haller:Futures}:
\begin{lstlisting}
promise (op x |-> original_interrupt_handler) as p in 
...
process$_{\op}$ p with (<<is>> |-> filter (fun i |-> i > 0) is) as q in 
process$_{\op}$ q with (<<js>> |-> fold (fun j j' |-> j * j') 1 js) as r in 
process$_{\op}$ r with (<<k>> |-> send productOfPositiveElements k) as _ in 
...
\end{lstlisting}
We note that for this to work, it is crucial that incoming interrupts behave 
like (deep) effect handling (see \autoref{sec:basic-calculus:semantics:computations})
so that all three post-processing computations get executed, in their program order.

% !TEX root = paper.tex

\section{Conclusion}
\label{sec:conclusion}

We have shown how to incorporate asynchrony within 
algebraic effects, by decoupling 
the execution of operation calls into signalling that an operation's implementation 
needs to be executed, and interrupting a running computation with the operation's result, 
to which it can react by installing interrupt handlers.
We have shown that our approach is flexible enough that not all signals have to have a matching 
interrupt, and vice versa, allowing us to also model spontaneous behaviour, such as a user 
clicking a button or the environment preempting a thread. We have formalised these ideas in a small 
calculus, called \lambdaAEff, and demonstrated its flexibility on a number of examples.
We have also accompanied the paper with an \pl{Agda} formalisation and a prototype implementation of \lambdaAEff.
However, various future work directions still remain. We discuss these and related work below.

\paragraph{Asynchronous effects}
As asynchrony is desired in practice, it is no surprise that \pl{Koka} \cite{Leijen:AsyncAwait} 
and \pl{Multicore OCaml} \cite{Dolan:MulticoreOCaml}, the two largest implementations of algebraic effects 
and handlers, have been extended accordingly. 
In \pl{Koka}, algebraic operations  
reify their continuation into an explicit callback structure that is then dispatched to a primitive 
such as \lstinline{setTimeout} in its \pl{Node.JS} backend. In \pl{Multicore OCaml}, one uses low-level functions 
such as \lstinline{set_signal} or \lstinline{timer_create} that modify the runtime by interjecting operation 
calls inside the currently running code. Both approaches thus \emph{delegate} the actual asynchrony to existing 
concepts in their backends. In contrast, in \lambdaAEff, we 
can express such backend features within the core calculus itself.

Further, in \lambdaAEff, we avoid having to manually use (un)masking to 
disable asynchronous effects in unwanted places 
in our programs, which can be a very tricky business to get right, as noted by \citet{Dolan:MulticoreOCaml}.
Instead, by design, interrupts in \lambdaAEff~\emph{never} 
influence running code unless the code has an explicit interrupt handler installed, 
and they \emph{always} wait for any potential handler to present itself during
execution (recall that they get discarded only when reaching a $\tmkw{return}$).

\paragraph{Message-passing}
While in this paper we have focussed on the foundations of asynchrony in the 
context of algebraic effects, the ideas we propose have also many common 
traits with concurrency models based on \emph{message-passing}, 
such as the Actor model \cite{Hewitt:Actors}, the $\pi$-calculus \cite{Milner:PiCalculus}, 
and the join-calculus \cite{FournetGonthier:JoinCalculus}, just to name a few.
Namely, one can view the issuing of a signal $\tmopout{op}{V}{M}$ as sending a message, 
and handling an interrupt $\tmopin{op}{W}{M}$ as receiving a message, along a channel
named $\op$. 
In fact, we believe that in our prototype implementation we could replace the semantics 
presented in the paper with an equivalent one based on shared channels
(one for each $\op$), to which the interrupt handlers could subscribe to.
Instead of propagating signals first out and then in, they would be sent directly to channels 
where interrupt handlers immediately receive them, drastically reducing the cost of communication.

Comparing \lambdaAEff~to the Actor model, we see that 
the $\tmrun M$ processes evolve in their own bubbles, and only communicate with other
processes via signals and interrupts, similarly to actors.
However, in contrast to messages not being required to be ordered 
in the Actor model, in our $\tmpar P Q$, the process $Q$ receives 
interrupts in the same order as the respective signals are issued by $P$ 
(and vice versa). This communication ordering could be relaxed by allowing 
signals to be hoisted out of computations from deeper than just the top level.
Another difference 
with actors is that \lambdaAEff-computations 
can react to interrupts only sequentially, and not by dynamically creating new parallel 
processes---first-class parallel processes and their dynamic creation is something 
we plan to address in future work.

It is worth noting that our interrupt handlers are similar to the message receiving construct 
in the $\pi$-calculus, in that they both synchronise with matching incoming
interrupts/messages. However, the two are also different, in that interrupt handlers allow
reductions to take place under them and non-matching interrupts to propagate past them.
Further, our interrupt handlers are also similar to join definitions in the join-calculus, describing
how to react when a corresponding interrupt arrives or join pattern appears, where in both cases
the reaction could involve effectful code. To this end, our interrupt handlers resemble join definitions 
with simple one-channel join patterns. However, where the two constructs differ is that join definitions 
also serve to define new (local) channels, similarly to the restriction operator in the $\pi$-calculus, 
whereas we assume a fixed global set of channels (i.e., signal and interrupt names). 
We expect that extending \lambdaAEff~with local algebraic effects 
\cite{Staton:Instances,Biernacki:AbstractingAlgEffects}
could help us fill this gap between the formalisms.

\paragraph{Scoped operations}
As discussed in \autoref{sec:basic-calculus:semantics:computations}, despite their name, interrupt handlers
behave like algebraic operations, not like effect handlers. However, one should also note 
that they are not conventional operations because they carry computational data that sequential 
composition does not interact with, and that only gets triggered when a corresponding interrupt is received. 

Such generalised operations are known in the literature as \emph{scoped operations}~\cite{Pirog:ScopedOperations},  
a leading example of which is $\opsym{spawn}(M;N)$, where $M$ is the new child process to be executed and $N$ is 
the current process. Crucially, the child $M$ should not directly interact with the current process. Scoped operations 
achieve this behaviour by declaring $M$ to be in the scope of $\opsym{spawn}$, resulting in  
$\tmlet x {\opsym{spawn}(M;N)} {K} \!\reduces\! \opsym{spawn}(M;\tmlet x N K)$, exactly 
as we have for interrupt handlers.

Further recalling \autoref{sec:basic-calculus:semantics:computations}, despite their appearance, 
incoming interrupts behave computationally like effect handling, not like algebraic operations. 
In fact, it turns out they correspond to effect handling 
induced by an instance of \emph{scoped effect handlers} \cite{Pirog:ScopedOperations}.
Compared to ordinary effect handlers, scoped effect handlers explain both 
how to interpret operations and their scopes. In our setting, this 
corresponds to selectively executing the handler code of interrupt handlers.

It would be interesting to extend our work both with 
scoped operations having more general signatures, and with additional effect handlers 
for them. The latter could allow  
preventing the propagation of incoming interrupts into continuations, discarding the continuation 
of a cancelled remote call, and techniques such as masking or reordering interrupts
according to priority levels.

\paragraph{Modal types}
We recall that the type safety of \lambdaAEff~crucially relies on 
the promise-typed variables bound by interrupt handlers not being
allowed to appear in the payloads of signals. This ensures that it is safe to propagate  
signals past all enveloping interrupt handlers, and communicate their payloads 
to other processes. In its essence, this is similar to the use of \emph{modal types} in distributed 
\cite{Murphy:PhDThesis} and reactive programming \cite{Krishnaswami:HOFRP,Bahr:RATT}
to classify values that can travel through space and time. In our case, 
it is the omission of promise types from ground types that allows us to consider  
the payloads of signals and interrupts as such \emph{mobile values}. 

We expect that these connections to modal types will be key for  
extending \lambdaAEff~with (i)~higher-order payloads and (ii) process 
creation. For (i), we want to avoid the bodies of function-typed payloads to be able 
to await enveloping promise variables to be fulfilled. For (ii), 
we want to do the same for the dynamically created processes.
In both cases, the reason is to be able to safely propagate the corresponding 
programming constructs past enveloping interrupt handlers, and eventually 
hoist them out of individual computations. We believe that the more structured 
treatment of contexts $\Gamma$, as studied in various modal type 
systems, will hold the key for these extensions to be type safe.

\paragraph{Denotational semantics}
In this paper we study only the operational side of \lambdaAEff, 
and leave developing its denotational semantics for the future.
In light of how we have motivated the \lambdaAEff-specific programming 
constructs, and based on the above discussions, we expect the denotational semantics 
to take the form of an algebraically natural \emph{monadic semantics}, where the monad would 
be given by an instance of the one studied by \citet{Pirog:ScopedOperations} for 
scoped operations (quotiented by the commutativity of signals and interrupt handlers, 
and extended with nondeterminism to model different evaluation
outcomes), incoming interrupts would be modelled as homomorphisms
induced by scoped algebras, and parallel composition 
by considering all nondeterministic interleavings of (the outgoing signals of) individual computations, e.g., 
based on how \citet{Plotkin:BinaryHandlers} and \citet{Lindley:DoBeDoBeDo} model 
it in the context of general effect handlers.
Finally, we expect to take inspiration for the denotational semantics of the promise type 
from that of modal logics and modal types.

\paragraph{Reasoning about asynchronous effects}
In addition to using \lambdaAEff's type-and-effect system only for specification purposes (such as specifying 
that $M : \tycomp{X}{(\emptyset,\{\})}$ raises no signals and installs no interrupt handlers), 
we wish to make further use of it for validating \emph{effect-dependent optimisations} \cite{Kammar:Optimisations}. 
For instance, whenever $M : \tycomp{X}{(\o,\i)}$ and $\i\, (\op) = \bot$, we would like to know  
that $\tmopin{\op}{V}{M} \reduces^* M$. One way to validate such optimisations 
is to develop an adequate denotational semantics, 
and then use a semantic \emph{computational induction} principle \cite{Bauer:EffectSystem,Plotkin:Logic}.
For \lambdaAEff, this would amount to only having to prove the optimisations for return values, signals, 
and interrupt handlers. Another way to validate effect-dependent optimisations would 
be to define a suitable logical 
relation for \lambdaAEff~\cite{Benton:AbstractEffects}.

In addition to optimisations based on \lambdaAEff's existing effect system, 
we plan to explore extending processes and their types 
with \emph{communication protocols} inspired by session types \cite{Honda:LangPrimitives}, 
so as to refine the current ``broadcast 
everything everywhere'' communication strategy.

\section*{Acknowledgements}

We thank the anonymous reviewers, Otterlo IFIP WG 2.1 meeting participants, 
and Andrej Bauer, Gavin Bierman, Žiga Lukšič, and Alex Simpson for their useful feedback.
%
%%% Funding
%
This project has received funding from the European Union’s Horizon 2020 research and 
innovation programme under the Marie Sk\l{}odowska-Curie grant agreement No 834146
\raisebox{-0.05cm}{
\hspace{-0.15cm}
\includegraphics[width=0.5cm]{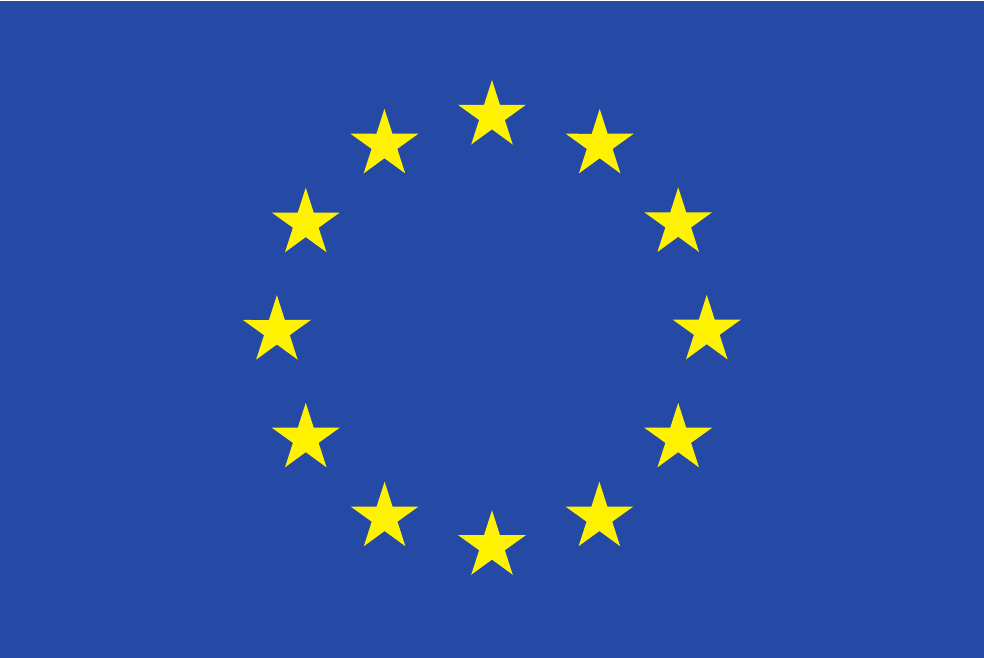}
\hspace{-0.15cm}
}.
This material is based upon work supported by the Air Force Office of Scientific 
Research under award number FA9550-17-1-0326.

%% BIBLIOGRAPHY
\bibliographystyle{ACM-Reference-Format}
\bibliography{references}

\typeout{get arXiv to do 4 passes: Label(s) may have changed. Rerun}
\end{document}